\RequirePackage{fix-cm}
\documentclass[smallextended]{svjour3}       
\smartqed  

\usepackage{latexsym}
\usepackage{bm,bbm}
\usepackage{amsmath,amsfonts,amscd,amssymb}
\usepackage{esint}
\usepackage{amstext}
\usepackage{graphicx}
\usepackage{verbatim}
\usepackage{graphicx,color}
\usepackage{upgreek}
\usepackage{url}
\usepackage{xcolor}
\numberwithin{equation}{section}
\usepackage{tikz}
\usetikzlibrary{decorations.markings}
\usepackage{tkz-graph}

 \usepackage{mathptmx}      
\newcommand{\be}{\begin{equation}}
\newcommand{\ee}{\end{equation}}
\newcommand{\ben}{\begin{equation}}
\newcommand{\een}{\end{equation}}

\newcommand{\barr}{\begin{eqnarray}}
\newcommand{\earr}{\end{eqnarray}}
\newcommand{\de}{\mathrm{d}}
\newcommand{\EE}{\mathcal{E}}
\newcommand{\R}{\mathbb{R}}

\newcommand{\C}{\mathbb{C}}

\newcommand{\Rs}{R_{\star}}

\newtheorem{assumptions}{Assumption}
\newenvironment{dimostrazione}{\removelastskip\par\medskip
\noindent{\em Proof of}
\emph}{\penalty-20\null\hfill$\square$\par\medbreak}

\spnewtheorem*{notation}{Notational remark}{\itshape}{\rmfamily}

\begin{document}

\title{Third-order phase transition: random matrices  and screened Coulomb gas with hard walls}


\author{Fabio Deelan Cunden \and
Paolo Facchi \and\\
 Marilena Ligab\`o \and
  Pierpaolo Vivo   
}


\institute{Fabio Deelan Cunden \at
              School of Mathematics and Statistics, University College Dublin, Dublin 4, Ireland \\
              \email{fabio.cunden@ucd.ie}           
           \and
           Paolo Facchi  \at
              Dipartimento di Fisica and MECENAS, Universit\`a di Bari, I-70126 Bari, Italy\\
              Istituto Nazionale di Fisica Nucleare (INFN), Sezione di Bari, I-70126 Bari, Italy
              \and
              Marilena Ligab\`o \at
              Dipartimento di Matematica, Universit\`a di Bari, I-70125 Bari, Italy
              \and
                Pierpaolo Vivo   \at
               King's College London, Department of Mathematics, Strand, London WC2R 2LS, United Kingdom
}

\date{Received: date / Accepted: date}

\maketitle

\begin{abstract}

Consider the free energy of a $d$-dimensional gas in canonical equilibrium under pairwise repulsive interaction and global confinement, in presence of a volume constraint. When the volume of the gas is forced away from its typical value,  the system undergoes a phase transition of the third order separating two phases (pulled and pushed).
We prove this result  i) for the eigenvalues of one-cut, off-critical random matrices (log-gas in dimension $d=1$) with hard walls; ii) in arbitrary dimension $d\geq1$ for a gas with Yukawa interaction (aka screened Coulomb gas) in a generic confining potential. The latter class includes systems with Coulomb (long range) and delta (zero range) repulsion as limiting cases. 
In both cases, we obtain an exact formula for the free energy of the constrained gas which explicitly exhibits a jump in the third derivative, and we identify the `electrostatic pressure' as the order parameter of the transition. 
Part of these results were announced in [F. D. Cunden, P. Facchi, M. Ligab\`o and P. Vivo, J. Phys. A: Math. Theor. {\bf 51}, 35LT01 (2018)].

\end{abstract}


\section{Introduction and statement of results}
\label{sec:intro}

Phase transitions -- points in the parameter space which are singularities in the free energy --  generically occur in the study of ensembles of random matrices, as the parameters in the joint probability distribution of the eigenvalues are varied~\cite{Cicuta01}. The aim of this paper is to characterise the \emph{pulled-to-pushed phase transition} (defined later) in random matrices with hard walls and, more generally, in systems with repulsive interaction in arbitrary dimensions.
\par
Given an interaction kernel $\Phi\colon\R^d\to(-\infty,+\infty]$ and a potential $V\colon\R^d\to\R$, we define the energy associated to a gas of $N$ particles at position $x_i\in\R^d$ ($i=1,\dots,N)$ as
\begin{equation}
E_N(x_1,\ldots,x_N)=\frac{1}{2}\sum_{i\neq j}\Phi(x_i-x_j)+N\sum_k V(x_k),\quad x_i\in\R^d\ . \label{energy1}
\end{equation}
We regard $-\nabla \Phi(x-y)$ as the force that a particle at $x$ exerts on a particle at $y$, and $V(x)$ as a global coercive potential energy, $V(x)\to+\infty$ as $|x|\to\infty$. The typical interactions we have in mind are repulsive at all distances, i.e. $-\nabla \Phi(x)\cdot x\geq0$. It is natural to assume that the repulsion is isotropic $ \Phi(x)=\varphi(|x|)$ and that the potential is radial $V(x)=v(|x|)$.  
\par
The normalisation of the energy is done in such a way that both terms (the sum over pairs and the sum of one-body terms) are of same order $\operatorname{O}(N^2)$ for large $N$. Indeed, in terms of the `granular' normalised particle density,
\begin{equation}
\label{eq:empmeas}
\rho_N = \frac{1}{N} \sum_i \delta_{x_i}\ ,
\end{equation}
the energy~\eqref{energy1} reads
\begin{equation}
E_N(x_1,\ldots,x_N) = N^2 \left[ \frac{1}{2}\iint_{x\neq y}\Phi(x-y)\de\rho_N(x)\de\rho_N(y)+\int  V(x)\de\rho_N(x) \right] \, .
\end{equation}
\par
The minimisers $\rho_N$ of the discrete energy should achieve the most stable balance between the repulsive effect of the interaction term and the global confinement. Finding global and constrained minimisers of the discrete energy $E_N$ is a question of major interest in the theory of optimal point configurations. 
\par
For a large class of interaction kernels, the sequence of minimisers $\rho_N$ converges toward some non-discrete  measure $\rho$ when $N\to\infty$.
It is therefore convenient to reframe  the optimisation problem in terms of a field functional $\EE\colon\mathcal{P}(\R^d)\to(-\infty,+\infty]$ defined on the set of probability measures $\rho\in\mathcal{P}(\R^d)$ by
\be
\EE[\rho]=\frac{1}{2}\iint_{\R^d\times\R^d}\Phi(x-y)\de\rho(x)\de\rho(y)+\int_{\R^d} V(x)\de\rho(x)\ ,
\label{eq:EF}
\ee
where $\rho(x)$ represents the normalised density of particles around the position $x\in\R^d$. 
\par
Mean field energy functionals of the form~\eqref{eq:EF} and their minimisers have received attention in the study of asymptotics of the partition functions of interacting particle systems.  
Consider the positional partition function of a particle system defined by the energy~\eqref{energy1} at inverse temperature $\beta>0$
\be
Z_N=\int e^{-\beta E_N}\de x_1\cdots \de x_N\ .
\ee
For several particle systems including eigenvalues of random matrices, Coulomb and Riesz gases, the leading term in the asymptotics of the free energy is the minimum of the mean field energy functional~\cite{Chafai14,Leble15,Rougerie16}
\be
-\frac{1}{\beta N^2}\log Z_N \to \min_{\rho\in\mathcal{P}(\R^d)}\EE[\rho]\ , \qquad \text{as}\quad N\to\infty\ .
\label{eq:asympt_Z}
\ee
\par
\begin{remark}
We stress that physically the change of picture $\{x_i\}\to\rho(x)$ is accompanied by an entropic contribution of the gas (the `number' of microstates $\{x_i\}$ of the gas that contribute to a given macroscopic density profile $\rho(x)$). However, the scaling in $N$ of the energy is such that the entropic contribution is always sub-leading in the large $N$ limit, and the free energy is dominated by the internal energy component. 
 To see this, we remark that  the Boltzmann factor corresponding to the energy~\eqref{energy1} can be written as
  \begin{equation}
  e^{-\beta E_N( x_1,\dots,x_N)}=\exp\left\{-\beta N\left(\frac{1}{2N}\sum_{i\neq j}\Phi(x_i-x_j)+\sum_k V(x_k)\right)\right\}\ .
  \label{eq:MF}
  \end{equation}
Therefore, the mean-field limit $N\to\infty$ is simultaneously a limit of large number of particles \emph{and} a zero-temperature limit, with an energy $\operatorname{O}(N)$ extensive in the number of particles as in standard statistical mechanics. (We learned this argument from a paper by Kiessling and Spohn~\cite{Spohn99}.) One then expects that in the large-$N$ limit, the free energy of $N$ particles at  inverse mean-field temperature $\beta_{MF}=\beta N$ approaches the minimum energy since at zero temperature the entropy is absent. This also explains the rescaling in~\eqref{eq:asympt_Z}, as $\beta_{MF}N=\beta N^2$. 
\end{remark}
Having clarified the physical meaning of the mean-field functional~\eqref{eq:EF}, we now turn to the problem addressed in this paper.

\begin{notation}
Throughout the paper $\mathcal{P}(B)$ denotes the set of probability measures whose support lies  in $B\subset\R^d$.  The   Euclidean ball of radius $R$ centred at $0$ is denoted by  $B_R=\{x\in \R^d\colon |x|=(x_1^2+\cdots +x_d^2)^{1/2}\leq R\}$. If $F$ is a formula, then $\mathbbm{1}_{F}$ is the indicator of the set defined by the formula $F$. 
We also use the notation $a\wedge b=\min\{a,b\}$.
\end{notation}
\subsection{Formulation of the problem}

Consider a particle systems with Boltzmann factor as in~\eqref{eq:MF} and assume that,  for large $N$, the partition function behaves like~\eqref{eq:asympt_Z}.
Under general hypotheses, the balance between mutual repulsion and external confinement allows for the existence of a \emph{compactly supported} global minimiser of the energy functional $\EE$ in~\eqref{eq:EF}. In radially symmetric systems ($\Phi(x)=\varphi(|x|)$ and $V(x)=v(|x|)$), the minimiser $\rho_{R_{\star}}$ is supported on a ball 
\ben
\EE[\rho_{R_{\star}}]=\inf_{\rho\in\mathcal{P}(\R^d)}\EE[\rho],\qquad \int_{B_{R_{\star}}}\de\rho_{R_{\star}}(x)=1\ .
\een
We say that the gas is at equilibrium in a ball of radius $R_{\star}$ with density $\rho_{\Rs}$.
\par
Suppose that we want to compute the probability that the gas is contained in a  volume $B_R$ with $R\neq R_{\star}$, i.e. 
\ben
\operatorname{Pr}\left(x_i\in B_R,\, i=1,\dots,N\right)=\frac{\int\limits_{x_i\in B_R}e^{-\beta E_N}\de x_1\cdots \de x_N}{\int\limits_{x_i\in \R^d}e^{-\beta E_N}\de x_1\cdots \de x_N}=\frac{Z_N(R)}{Z_N(\infty)}\ .
\label{eq:prob_conf}
\een
The denominator is nothing but the partition function of the gas $Z_N(\infty)=Z_N$, while the integral in the numerator is the partition function $Z_N(R)$ of the same gas constrained to stay in the ball $B_R$. 
In view of the asymptotics~\eqref{eq:asympt_Z}, for large $N$
\ben
\log Z_N(R)=-\beta N^2\EE[\rho_R]+  o(N^2) \ ,
\een
where $\rho_R$ is the equilibrium measure of the gas confined in $B_R$
\ben
\EE[\rho_R]=\inf_{\rho\in\mathcal{P}(B_R)}\EE[\rho]\ .
\label{eq:constmin}
\een
(Note that $\rho_R = \rho_{R_{\star}}$ for all $R\ge R_{\star}$.) We conclude that  the probability~\eqref{eq:prob_conf} decays as
\ben
\operatorname{Pr}\left(x_i\in B_R,\, i=1,\dots,N\right)\approx e^{-\beta N^2 F(R)}\ ,
\een
where the \emph{large deviation function} $F(R)$ is 
\begin{align}
F(R)
=- \lim_{N\to\infty}\frac{1}{\beta N^2}\left(\log Z_N(R)-\log Z_N(\infty)\right)=\EE[\rho_R]-\EE[\rho_{R_\star}]\ .
\label{eq:def_F}
\end{align}
Clearly, $F(R)\geq0$ and is non-increasing. 
The physical interpretation of $F(R)$ and $\rho_R$ is clear: $F(R)$ is the \emph{excess free energy} of the gas, constrained within $B_R$, with respect to the situation where it occupies the unperturbed volume $B_{R_\star}$; the measure $\rho_R$ describes the equilibrium density of the constrained gas in the limit of large number of particles~$N$.
\par
The general picture is as follows (see Fig.~\ref{fig:scheme}): 
\begin{itemize}
\item[i)] In the unconstrained problem ($B_R=\R^d$) the global minimiser $\rho_{R_{\star}}$ is supported on the ball $B_{R_{\star}}$;
\item[ii)] If $R>R_{\star}$, the constraint in~\eqref{eq:def_F} is immaterial ($B_R$ contains $B_{R_\star}$), and hence the equilibrium measure is $\rho_R=\rho_{R_{\star}}$ and $F(R)=0$. This is the so-called \emph{pulled phase}, borrowing a terminology suggested in~\cite{Nadal11}; 
\item[iii)] If $R<R_{\star}$ the system is in a \emph{pushed phase}, the constraint is effective, and the equilibrium energy of the system increases $\EE[\rho_R]\geq\EE[\rho_{R_\star}]$. 
\item[iv)] At $R=R_{\star}$ the gas undergoes  a \emph{phase transition} and the free energy $F(R)$ displays a non-analytic behaviour. Typically at microscopic scales one expects a crossover function separating the pushed and pulled phases.
\end{itemize}
\begin{figure}
\centering
  \includegraphics[width=.325\columnwidth]{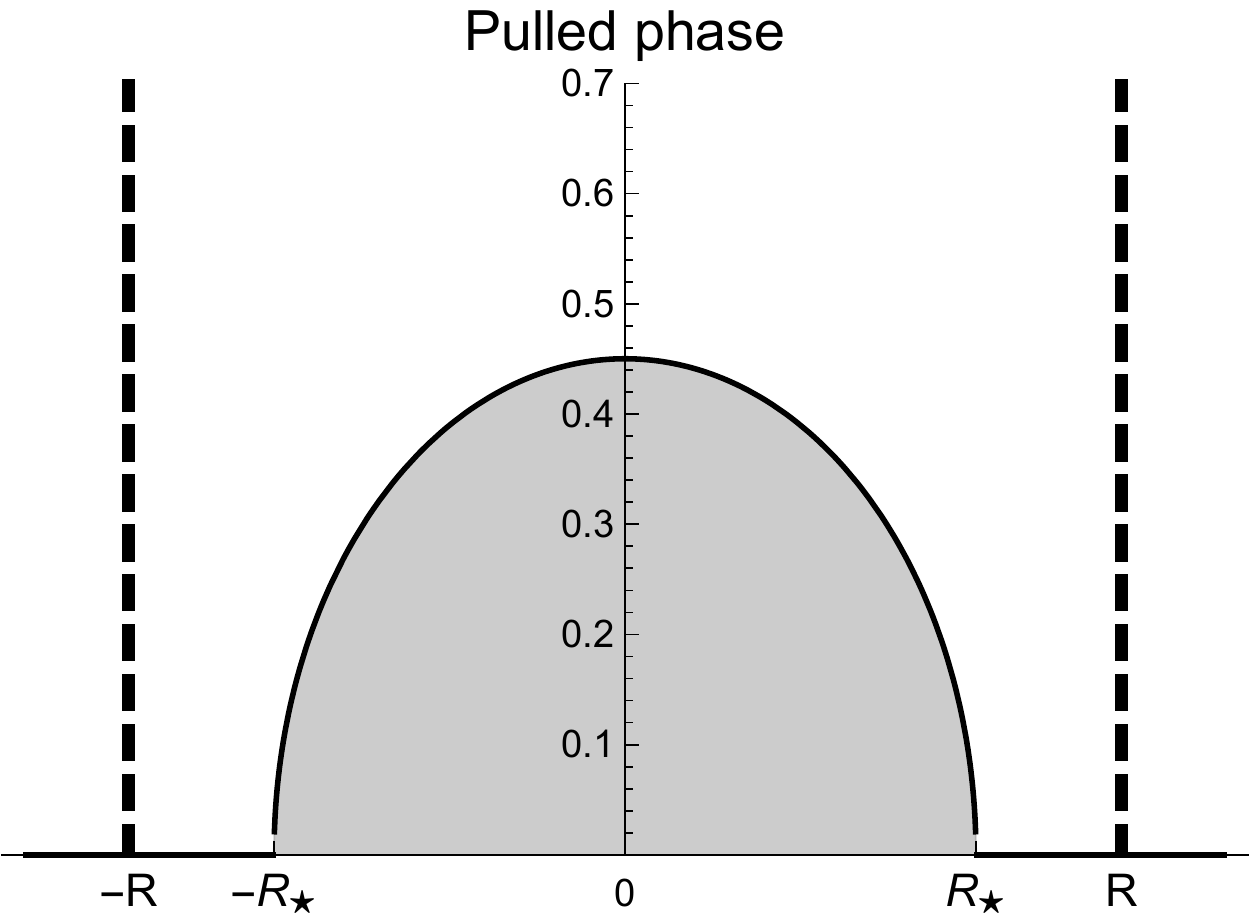}
    \includegraphics[width=.325\columnwidth]{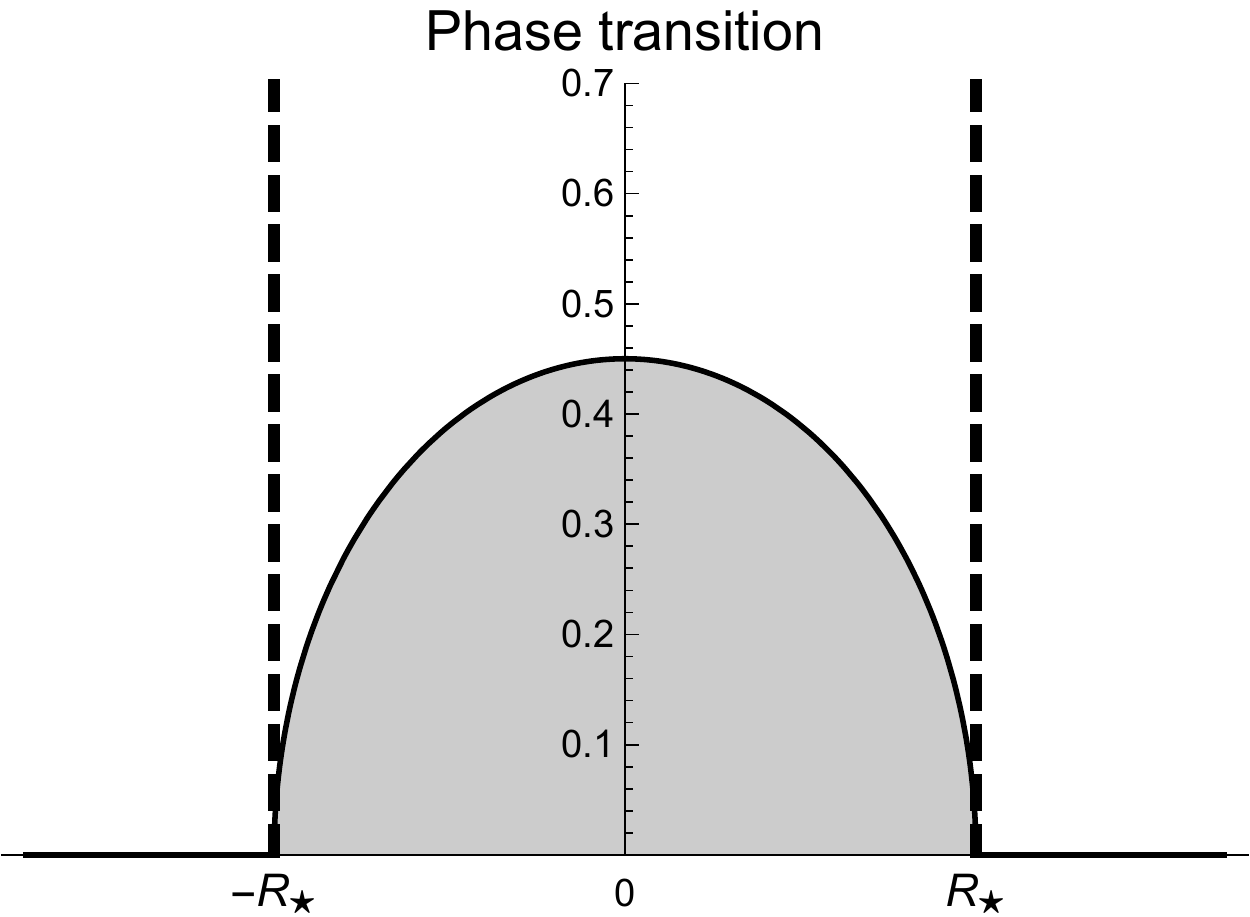}
      \includegraphics[width=.325\columnwidth]{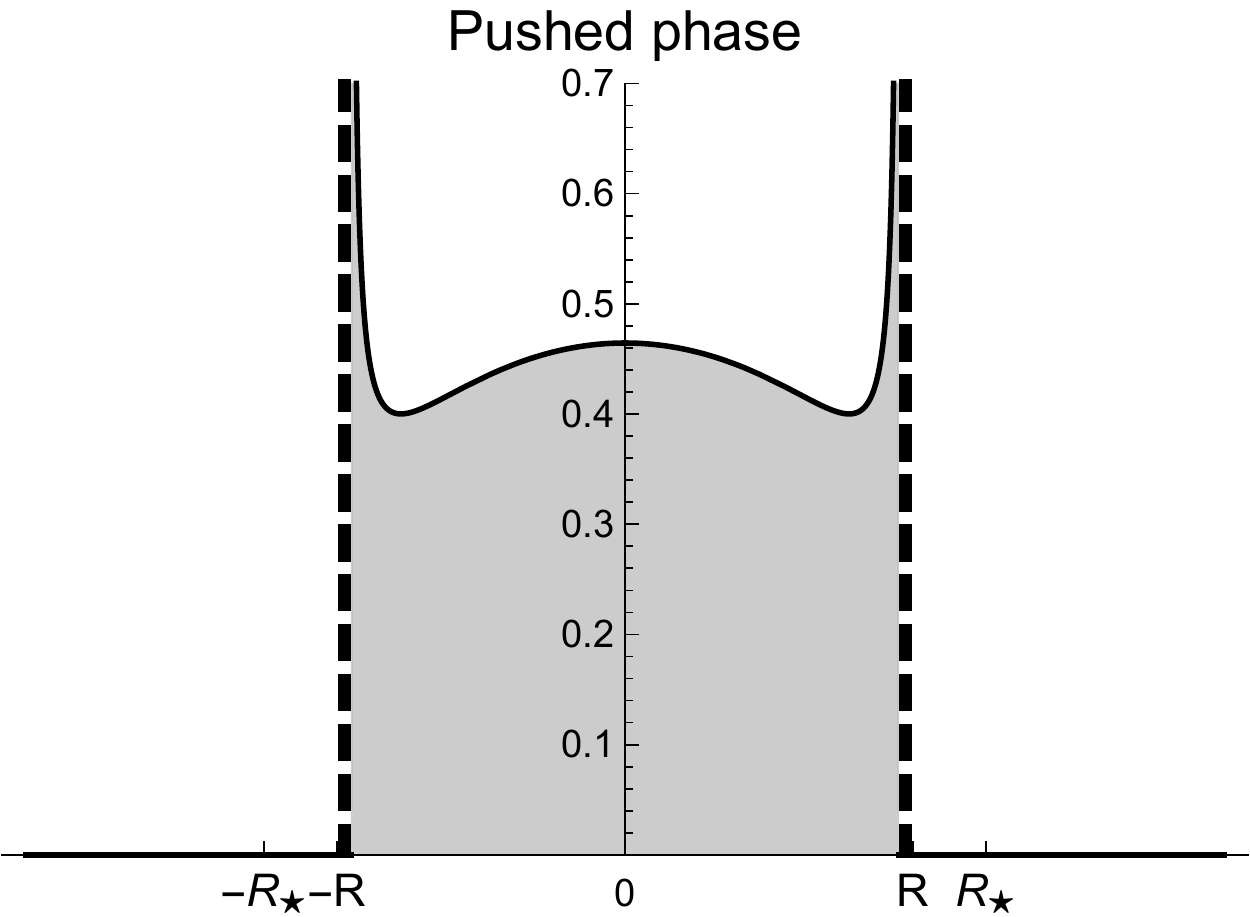}
\caption{ The pulled-to-pushed transition for a log-gas in dimension $d=1$ in a quadratic potential (GUE).}
\label{fig:scheme}       
\end{figure}
\par
The goal of this work is to investigate the properties of the excess free energy
\be
F(R)=\inf_{\rho\in\mathcal{P}(B_R)}\EE[\rho]-\EE[\rho_{R_\star}]
\label{eq:def_F_gen}
\ee
at the critical point $R_\star$ for certain systems with pairwise repulsive interactions. Explicitly solvable models related to random matrices suggest that in the vicinity of the critical point
\be
F(R)\simeq(R_{\star}-R)^3\mathbbm{1}_{R\leq R_{\star}}\ ,
\label{eq:F_vicinity}
\ee
implying that the transition between the \emph{pushed} and \emph{pulled} phases of the gas is third-order. Here we demonstrate that~\eqref{eq:F_vicinity} is generically true for a large class of systems with repulsive interactions. More precisely, we prove~\eqref{eq:F_vicinity}:
\begin{itemize}
\item[i)] For the log-gas $\Phi(x)=-\log|x|$ in  dimension $d=1$ (eigenvalues of one-cut, off-critical random matrices with hard walls);
\item[ii)] For the Yukawa gas $\Phi(x)=\Phi_d(x)$ with
\be
\Phi_d(x)=\frac{1}{a^2 2^{\frac{d}{2}-1}}\frac{1}{\Gamma\left(\frac{d}{2}\right)}\left(\frac{m}{a|x|}\right)^{\frac{d}{2}-1}K_{\frac{d}{2}-1}\left(\frac{m|x|}{a}\right)
\label{eq:Yukawa_kernel}
\ee
in arbitrary dimension $d\geq1$, including its limiting cases $m\to0$ (Coulomb gas) and  $a\to0$ (Thomas-Fermi gas). In~\eqref{eq:Yukawa_kernel}, $K_\nu$ denotes the modified Bessel function of the second kind.
 \end{itemize}
 \par
 The precise assumptions on the confining potential $V(x)$ are presented together with the statements of Theorem~\ref{thm:log-gas2} and \ref{thm:Yukawa} below. 

\subsection{Extreme eigenvalues of random matrices and log-gases with hard walls}
\subsubsection{Hermitian random matrices}
Many random matrix theory (RMT) phenomena have been first discovered for invariant measures on the space of $N\times N$ complex Hermitian matrices $M$ of the form
\be
\de P(M)=\dfrac{\displaystyle e^{-\beta N\operatorname{Tr}V(M)}\de M}{\displaystyle\int e^{-\beta N\operatorname{Tr}V(M')}\de M'}\ ,
\label{eq:Pcanonical}
\ee
where $V$ is a scalar function referred to as the potential of the matrix model and $\beta=2$. 

Expectations of conjugation-invariant random variables with respect to these measures can be reduced, via the Weyl denominator formula, to an integration against the joint density of the eigenvalues $x_1,\dots,x_N$, which has the form 
\begin{align*}
p_N(x_1,\dots,x_N)&=\frac{1}{Z_N}e^{-\beta E_N(x_1,\dots,x_N)}\nonumber\\
E_N(x_1,\dots,x_N)&=-\frac{1}{2}\sum_{i\neq j}\log|x_i-x_j|+N\sum_k V(x_k),\quad x_i\in\R\ .
\label{eq:Pcanonical2}
\end{align*}
This energy is of the form~\eqref{energy1} in dimension $d=1$ and with $\Phi(x)=-\log|x|$. Hence, $Z_N=\int_{\R}e^{-\beta E_N}\de x_1\cdots \de x_N$ can be interpreted as the  partition function of a \emph{log-gas} (or a $2d$ Coulomb gas) of $N$ repelling particles on the line in a global confining potential $V$. This physical interpretation also suggests that we may consider generic values of the \emph{Dyson index} $\beta>0$, interpreted as inverse temperature. 

In the large-$N$ limit the eigenvalue empirical measure $\rho_N=\frac{1}{N}\sum_i\delta_{x_i}$  weakly converges to a deterministic density. This limit is the \emph{equilibrium measure} (the minimiser) of the functional 
\be
\EE[\rho]=-\frac{1}{2}\iint_{\R\times\R}\log|x-y|\de\rho(x)\de\rho(y)+\int_{\R} V(x)\de\rho(x)\ . 
\label{eq:EF_log}
\ee

For concreteness, let us focus on the Gaussian Unitary Ensemble (GUE) defined by the measure~\eqref{eq:Pcanonical} with $V(x)=x^2/2$ and $\beta=2$. 
In this case, the equilibrium measure is supported on the symmetric interval $[-R_{\star},R_{\star}]$ where $R_{\star}=\sqrt{2}$, with a  density named after Wigner (semicircular law)
\ben
\de \rho_N\to\frac{1}{\pi}\sqrt{2-x^2}\mathbbm{1}_{|x|\leq R_{\star}} \de x\ .
\een
Moreover, as $N\to\infty$, the extreme statistics $\max|x_i|$ converges to the edge $R_{\star}$, namely $\operatorname{Pr}\left(\max|x_i|\leq R\right)$ converges to a step function: $0$ if $R<R_{\star}$, and $1$ if $R>R_{\star}$. For large $N$, the fluctuations of the spectral radius $\max|x_i|$ around $R_{\star}$ at the typical scale $\operatorname{O}(N^{-2/3})$ are described by a \emph{squared Tracy-Widom distribution}. In formulae~\cite{Dean16,Edelman15},
\ben
\lim_{N\to\infty}\operatorname{Pr}\left(\max|x_i|\leq R_{\star}+\frac{t}{\sqrt{2}N^{2/3}}\right)=\mathcal{F}_2^2(t)\ ,
\een
where $\mathcal{F}_{\beta}(t)$ is known as the $\beta$-Tracy-Widom distribution~\cite{Tracy94} and can be expressed in terms of the Hastings-McLeod solution of the Painlev\'e II equation. The macroscopic (atypical) fluctuations of $\max|x_i|$ are instead described by a large deviation function. More precisely, for all $\beta>0$ the following limit exists
\ben
- \lim_{N\to\infty}\frac{1}{\beta N^2}\log\operatorname{Pr}\left(\max|x_i|\leq R\right)=-\lim_{N\to\infty} \frac{1}{\beta N^2}\log\frac{Z_N(R)}{Z_N(\infty)}=F(R)\ .
\een

For the log-gas in $d=1$, the asymptotics of the partition function 
\ben
\log Z_N(R)=-\beta N^2\EE[\rho_R]+o(N^2),\quad\text{with}\quad \EE[\rho_R]=\inf_{\rho\in\mathcal{P}(B_R)}\EE[\rho]\ ,
\een
has been rigorously established in several works.
When $R\geq R_{\star}$, the equilibrium density is supported on the ball of finite radius $R_{\star}$, so that $\rho_{R}=\rho_{R_{\star}}$  (the volume constraint is ineffective).  
Hence, the  large deviation function  is the excess free energy~\eqref{eq:def_F_gen} of the gas of eigenvalues forced to stay between two hard walls at $\pm R$.
It is clear that $F(R)=0$ for $R\geq R_{\star}$ (pulled phase), while $F(R)\geq0$ for $R<R_{\star}$ (pushed phase) when the log-gas gets pushed by the hard walls at $\pm R$.
\par
The calculation of $F(R)$ for the GUE and its $\beta>0$ extensions was performed in detail by Dean and Majumdar~\cite{Dean06,Dean08} who found explicit expressions for the  density
\be
\rho_R(x)=
\begin{cases}
\displaystyle\frac{1}{\pi}\frac{2+R^2-2x^2}{2\sqrt{R^2-x^2}}\mathbbm{1}_{|x|< R}  &  \text{if $R< R_{\star}$ (pushed phase)}\smallskip\\
\displaystyle\frac{1}{\pi}\sqrt{2-x^2}\mathbbm{1}_{|x|\leq R_{\star}} & \text{if $R\geq R_{\star}$ (pulled phase)}\ ,
\end{cases}
\label{eq:density_DM}
\ee 
and for the excess free energy 
\begin{equation}
F_{\text{GUE}}(R)=
\begin{cases}
\displaystyle\frac{1}{32}\left(8R^2-R^4-16\log R-12+8\log2\right)& \text{if $R< R_{\star}$}\smallskip \\
0  & \text{if $R\geq R_{\star}$}\ .
\end{cases}
\label{eq:rateFGUE}
\end{equation}
A closer inspection of the latter formula provides a thermodynamical characterisation of the pulled-to-pushed transition. Indeed, we see that 
\begin{equation}
F_{\text{GUE}}(R)\sim \frac{\sqrt{2}}{3}(R_{\star}-R)^3\mathbbm{1}_{R\leq R_{\star}}\, ,
\end{equation}
as $R\to R_{\star}$. Therefore, the third derivative of the free energy of the log-gas at the critical point $R_{\star}=\sqrt{2}$ is discontinuous. 

Similar phase transitions of the pulled-to-pushed type have been observed in several physics models related to random matrices~\cite{Majumdar14,Cunden16s}, including large-$N$ gauge theories~\cite{Gross80,Wadia80,Santilli18,Barranco14}, longest increasing subsequences of random permutations~\cite{Johansson98}, quantum transport fluctuations in mesoscopic conductors~\cite{VMB08,VMB10,Cunden15,Grabsch15,Grabsch16}, non-intersecting Brownian motions~\cite{SMCF,FMS11}, entanglement measures in a bipartite system~\cite{Facchi08,Facchi10,NMV10,Facchi13}, random tilings~\cite{Colomo13,Colomo15},  random landscapes~\cite{Fyodorov12}, and the tail analysis in the KPZ problem~\cite{Krajenbrink19}. (See also the recent popular science articles~\cite{Buchanan14,Wolchover14}.) 

An explanation of the critical exponent `3' has been put forward by Majumdar and Schehr~\cite{Majumdar14} (see also~\cite{Atkin14}) based on a standard extreme value statistics criterion and a matching argument of the large deviation function behaviour in the vicinity of the critical value $R_{\star}$ and the left tail of the Tracy-Widom distribution~\cite{Ramirez06}
\begin{equation}
\mathcal{F}_{\beta}'(x)\approx \exp\left(-\frac{\beta}{24}|x|^3\right), \quad x\to-\infty\ .
\end{equation}
The criterion predicts that if the equilibrium density of a log-gas in the pulled phase vanishes as $\rho_{R_{\star}}(x)\sim\sqrt{R_{\star}^2-x^2}$ at the edges --  the so-called \emph{off-critical} case -- then the pulled-to-pushed phase transition is of the third order. This conjectural relation between the particular behaviour of the gas density and the arising non-analyticities in the free energies
has been verified in several examples, even though each particular case (i.e. each matrix ensemble defined by a potential $V$) requires working out explicitly the model-dependent $F(R)$ to compute the critical exponent.
\par
In this paper \emph{we derive a general explicit formula for the free energy $F(R)$ of a log-gas in dimension $d=1$ in presence of hard walls, and we prove the universality of the third-order phase transition for one-cut, off-critical matrix models}. This proves the prediction arising from the extreme value statistics criterion formulated in~\cite{Majumdar14}.
\par 
While here we examine problems with radial symmetry in both the potential and the hard walls as they constitute a paradigmatic framework and allow for a systematic treatment,  
the electrostatic interpretation we present in Section~\ref{subsec:log} below enjoys a wider range of application. In Remark~\ref{rmk:nonsym}, indeed, we will show how the formulae and conclusions concerning the symmetric case carry over without extra efforts to non-symmetric potentials or random matrices with a single hard-wall as well. 
\par
The constrained minimisation problem for the log-gas in $d=1$ is usually solved by using complex-analytic methods or Tricomi's formula, which are well-known to those working in potential theory and random matrices. Nevertheless, we found a particularly convenient (and perhaps not so well-known) method based on the decomposition into Chebyshev polynomials which is well-suited to this class of problems. At the heart of the method is the following pointwise \emph{multipole expansion} of the two-dimensional Coulomb interaction (proven in Appendix \ref{app:lemmas})
\be
-\log |x-y|=\log 2+\sum_{n\geq 1}\frac{2}{n}T_n(x)T_n(y)\qquad x,y\in[-1,1],\;  x\neq y\ ,
\label{eq:formula:Cheb}
\ee
where the $T_n$'s are the Chebyshev polynomials of the first kind. They are defined by the orthogonality relation
\be
\int_{-1}^1 \frac{T_n(x)T_m(x)}{\sqrt{1-x^2}}\de x=\delta_{nm}h_n\quad\text{with}\quad h_n =
\begin{cases}
\pi &\text{if $n=m=0$}\\
\pi/2 &\text{if $n=m\geq 1$}
\end{cases},
\label{eq:Cheby_orth}
\ee
and they form a complete basis of $L^2([-1,1])$ (with respect to the arcsine measure).
The above identity was recently used and discussed in~\cite{FKS16,Garoufalidis03} (the authors refer to some unpublished lecture notes by U. Haagerup).

We consider  potentials $V(x)$ satisfying the following assumptions. 
 \begin{assumptions}
\label{ass:potential_log-gas}
$V(x)$ is $C^{3}(\R)$,  symmetric $V(x)=V(-x)$, strictly convex and satisfies $\liminf_{|x|\to\infty}\frac{V(x)}{\log|x|}>1$.
\end{assumptions}
We remark that strictly convex and super-logarithmic $V(x)$'s are in the class of \emph{one-cut, off-critical} potentials. 
 
We are going to present and discuss two theorems for the log-gas that will be proven in Sec.~\ref{sec:log-gas}.

\begin{theorem}
\label{thm:EM_log-gas} 
In dimension $d=1$, let $\Phi(x)=-\log|x|$, and $V(x)$ be a potential satisfying Assumption~\ref{ass:potential_log-gas}. Then 
\begin{itemize}
\item[i)]
there exists a unique probability measure that is solution of the constrained minimisation problem~\eqref{eq:constmin} for the energy functional~\eqref{eq:EF}, and it takes the form
\begin{equation}
\de\rho_R(x)=
\begin{cases}
\displaystyle\frac{1}{\pi}\frac{P_R(x)}{\sqrt{R^2-x^2}}\mathbbm{1}_{|x|< R}\,\de x  &  \text{if $R< R_{\star}$ (pushed phase)}\smallskip\\
\displaystyle\frac{1}{\pi}Q(x)\sqrt{R_{\star}^2-x^2}\mathbbm{1}_{|x|\leq R_{\star}}\,\de x  & \text{if $R\geq R_{\star}$ (pulled phase)}\ ,
\end{cases}
\label{eq:sol_log-gas}
\end{equation}
where  $P_R(x)$ and $Q(x)= \lim_{R\uparrow \Rs}P_{R}(x)/(R^2-x^2)$ are nonnegative on the support $[-R,R]$ and $[-\Rs,\Rs]$, respectively.

\item[ii)] An explicit expression of $P_R(x)$ (and $Q(x)$) is as follows.
Denote by $c_n(R)$ the Chebyshev coefficients of $V(R x)$. i.e.,
\ben
c_n(R)=\frac{1}{h_n}\int_{-1}^1  \frac{V(R x)T_n(x)}{\sqrt{1-x^2}}\de x\ .
\label{eq:Cheb_coeff}
\een
Then, the equilibrium measure~\eqref{eq:sol_log-gas} is uniquely determined as
\be
P_R(x)=1-\sum_{n\geq1}nc_n(R)T_{n}(x/R)\ .
\label{eq:sol_log-gas2}
\ee
The critical radius $\Rs$ is the smallest positive solution of the equation
\be
\sum_{n\geq1}nc_{n}(\Rs)=1\ .
\label{eq:radius_log}
\ee
\end{itemize}
\end{theorem}

In other words,  Eq.~\eqref{eq:sol_log-gas} shows that in the pulled phase the equilibrium density $\rho_R(x)=\rho_{\Rs}(x)$  is supported on a single interval on the real line (\emph{one-cut property}), is strictly positive in the bulk, and vanishes as a square root at the edges $\pm\Rs$ (\emph{off-critical case}). Moreover, since $Q(\Rs)>0$ one gets
 \be P_{\Rs}(\Rs)=0 \quad \text{ but \quad$P_{\Rs}'(\Rs)\neq0$}\ ,
 \ee 
a fact that will be used later on.

On the other hand in the pushed phase, the density is strictly positive in its support and has an integrable singularity at the hard walls $\pm R$ (cf. with the GUE case~\eqref{eq:density_DM}). See Fig.~\ref{fig:scheme}.

\par
A corollary of formulae~\eqref{eq:sol_log-gas}-\eqref{eq:sol_log-gas2} is the first main result of this paper.
\begin{theorem}
\label{thm:log-gas2}
With the assumptions of Theorem~\ref{thm:EM_log-gas}, the excess free energy~\eqref{eq:def_F_gen} of the log-gas is
\be
F(R)=\frac{1}{2}\int_{R\wedge\Rs}^{\Rs}\frac{P_r(r)^2}{r}\de r\ .
\label{eq:simple_log}
\ee
Moreover, $F(R)$ displays the following non-analytic behaviour at $\Rs$:
\ben
F(R) >0, \quad \text{for  $R<  R_{\star}$}, \qquad  F(R) =0, \quad \text{for  $R\geq  R_{\star}$}\ ,
\label{eq:non-analyt_log}
\een
and
\be
F(R)\sim C_{\star}(R_{\star}-R)^3\mathbbm{1}_{R\leq R_{\star}}\ , \quad \text{as} \quad R\to \Rs\ ,
\label{eq:F_vicinity2}
\ee
with $C_{\star}>0$, that is 
\ben
F(\Rs)=F'(\Rs)=F''(\Rs)=0,\quad \text{but}\quad
\lim_{R\uparrow  R_{\star}}F'''(R) = - 3! C_{\star} <0\ .
\een
\end{theorem}
\par
 \begin{remark} In the statement of the results, the family of strictly convex potentials is not the largest class of potentials where the above result holds true. 
 What is really required is that  in the pulled phase  the associated equilibrium measure is supported on a single interval, and vanishes as a square root at the endpoints. (An explicit characterisation of these conditions is quite complicated.)  More precisely, 
Theorem~\ref{thm:EM_log-gas} is true when $V(x)$ is a one-cut  potential. `One-cut' means that the equilibrium measure is supported on a single bounded interval.  Theorem~\ref{thm:log-gas2} is valid under the hypotheses that $V(x)$ is one-cut and off-critical potential. `Off-critical' means that $P_{\Rs}(\Rs)=0$ but $P_{\Rs}'(\Rs)\neq0$ so that $\rho_{R_{\star}}(x)\sim\sqrt{R_{\star}^2-x^2}$ at the edges. This is certainly true for strictly convex potentials. For `critical' potentials the pulled-to-pushed transition is \emph{weaker} than third-order (see Eq.~\eqref{eq:3rd_log-gas_der} in the proof). These potentials are however `exceptional' in the `one-cut' class. The problem for `multi-cut' matrix models remains open.
 \end{remark}
 \par
 \begin{example} We reconsider the free energy of the GUE (the log-gas on the real line in a quadratic potential $V(x)=x^2/2$). There are only two nonzero Chebyshev coefficients~\eqref{eq:Cheb_coeff}
\be
c_0(R)=c_2(R)=\frac{R^2}{4}\ .
\ee 
From~\eqref{eq:radius_log} we read that the critical radius $\Rs$ is the positive solution of $R^2/2=1$, i.e., $\Rs=\sqrt{2}$. From~\eqref{eq:sol_log-gas2}, in the pushed phase $P_R(x)=(2+R^2-2x^2)/2$, so that $P_r(r)=(2-r^2)/2$. Applying the general formula~\eqref{eq:simple_log}, one easily computes the large deviation function as
\be
F_{\text{GUE}}(R)=\frac{1}{2}\int_{R}^{\sqrt{2}}\frac{\left(2-r^2\right)^2}{4r}\de r=\frac{1}{32}\left(8R^2-R^4-16\log R-12+8\log2\right)
\ee
for $R\leq\sqrt{2}$, and zero otherwise. This coincides with the known result~\eqref{eq:rateFGUE}.
\end{example}
\par
Theorem~\ref{thm:log-gas2} and the previous discussion might lead to conclude that the universality of the third-order phase transition is inextricably related to the presence of a Tracy-Widom distribution separating the pushed and pulled phases. A hint that this is not the case comes from the study of extreme statistics of non-Hermitian matrices whose eigenvalues have density in the complex plane.
\subsubsection{Non-Hermitian random matrices} 
\par 
Consider a log-gas in dimension $d=2$ (the `true' $2d$ Coulomb gas in the plane)
\begin{align}
E_N(x_1,\dots,x_N)&=-\frac{1}{2}\sum_{i\neq j}\log|x_i-x_j|+N\sum_k V(x_k),\quad x_i\in \R^2\ .
\label{eq:Pcanonical3}
\end{align}
The gas density for large $N$ converges to the equilibrium measure of the energy functional~\eqref{eq:EF_log} extended to measures supported on the complex plane. When $V(x)=|x|^2/2$ and $\beta=2$, with the identification $\C\simeq \R^2$, the model corresponds to the density of eigenvalues of the complex Ginibre (GinUE) ensemble of random matrices, a non-Hermitian relative of the GUE. The equilibrium measure in this case is uniform in the unit disk (circular law)
\ben
\de \rho_N\to\frac{1}{\pi}\mathbbm{1}_{|x|\leq R_{\star}} \de x\ ,
\een
with $R_{\star}=1$, but the typical fluctuations of the extreme statistics $\max|x_i|$ are \emph{not} in the Tracy-Widom universality class. More precisely, setting $\gamma_N=\log{N}-2\log\log{N}-\log{2\pi}$, Rider~\cite{Rider03} proved that
\ben
\lim_{N\to\infty}\operatorname{Pr}\left(\max|x_i|\leq R_{\star}+\frac{\gamma_N+t}{\sqrt{4N\gamma_N}}\right)=G(t)\ ,
\een
where the limit is the Gumbel distribution $G(t)=\exp(-\exp(-t))$. This result is universal~\cite{Chafai14b} for the log-gas in the plane at inverse temperature $\beta=2$. The atypical fluctuations are described by a large deviation function $F_{\text{GinUE}}(R)$ which can be computed by solving the constrained variational problem of a log-gas in the plane. For the Ginibre ensemble, this was computed by Cunden, Mezzadri and Vivo~\cite{Cunden16} who found\footnote{This formula is also implicit in the work of Allez, Touboul and Wainrib~\cite{Allez14}.}
\be
\de \rho_R(x)=
\begin{cases}
\displaystyle\frac{1}{\pi}\mathbbm{1}_{|x|\leq R} \de x
+(1-R^2) \frac{\delta(|x|-R)}{2\pi R}  &  \text{if $R< R_{\star}$ (pushed phase)}\smallskip\\
\displaystyle\frac{1}{\pi}\mathbbm{1}_{|x|\leq R_{\star}} \de x & \text{if $R\geq R_{\star}$ (pulled phase)}\ ,
\end{cases}
\ee
\ben
F_{\text{GinUE}}(R)=
\begin{cases}
\displaystyle\frac{1}{8}(4R^2-R^4-4\log R-3) &  \text{if $R< R_{\star}$}\smallskip \\  
0 &   \text{if $R\geq R_{\star}$} \ .
\end{cases}
\label{eq:rateFGinUE}
\een
Of course, the explicit form of the large deviation function is specific to the model (for instance $F_{\text{GUE}}(R)\neq F_{\text{GinUE}}(R)$). The surprising fact is that even in dimension $d=2$, the pushed-to-pulled transition of the log-gas is of the third order
\ben
F_{\text{GinUE}}(R)\sim\frac{4}{3}(R_{\star}-R)^3 \mathbbm{1}_{R \leq R_{\star}},\quad\mbox{as } R\to R_{\star}\ .
\een
We remark that, in this case, the order of the phase transition is not predicted by the classical `matching argument'. 
In fact, for the log-gas in $d=2$, the matching between the typical fluctuations (Gumbel) and the large deviations is more subtle due to the presence of an intermediate regime, as found recently in~\cite{Lacroix17,Lacroix18}.
This suggests that the critical exponent `3' is shared by systems with repulsive interaction whose microscopic statistics belongs to different universality classes\footnote{For the $1d$ Coulomb gas on the line, the pulled-to-pushed transition is  of the third-order, despite the fact that the typical fluctuations of extreme particles are neither Tracy-Widom nor Gumbel. See~\cite{Dhar17,Dhar18}.}.  For which interactions can Theorem~\ref{thm:log-gas2} be extended?

\subsection{Beyond random matrices: Yukawa interaction in arbitrary dimension}
The logarithmic interaction $\Phi(x-y)=-\log|x-y|$ is the Coulomb potential in dimension $d=2$, namely the Green's function of the Laplacian on the plane
\ben
-\Delta(-\log|x|)=2\pi\delta(x),\quad x\in\R^2\ .
\een
Based on this observation, in~\cite{Cunden17} we put forward the idea of considering energy functionals whose interaction potential $\Phi$ is the Green's function of some differential operator $\mathrm{D}$:
\ben
\mathrm{D}\Phi(x)=\Omega_d\delta(x),\quad x\in\R^d\ , 
\label{eq:Green}
\een
where the constant $\Omega_d=2\pi^{d/2}/\Gamma\left(d/2\right)$ is the surface area of the unit sphere in $\R^d$ ($\Omega_1=2$, $\Omega_2=2\pi$, $\Omega_3=4\pi$, etc.) 
\subsubsection{Coulomb and Thomas-Fermi gas}
When $\mathrm{D}=-\Delta$ in $\R^d$, the system corresponds to a $d$-dimensional Coulomb gas (a system with long range interaction)
\ben
\Phi(x)=
\begin{cases}
\dfrac{1}{(d-2)}\dfrac{1}{|x|^{d-2}} &\text{if $d\neq 2$}\\
-\log |x|   &\text{if  $d=2$}\ . 
\end{cases}\qquad
\text{(Coulomb)}
\label{eq:Coulombdef}
\een
The constrained variational problem for a $d$-dimensional Coulomb gas (in $\R^d$) can be solved using macroscopic electrostatic considerations. 
\par
Another solvable model is the gas with delta potential corresponding to $\mathrm{D}=I$ (the identity operator) in $\R^d$. In this case, the repulsive kernel is proportional to a delta function (zero range interaction)  
\ben
\Phi(x)=\Omega_d\delta(x),\quad x\in\R^d \ , 
\label{eq:ThomasFermidef}
\een
and the energy  
\ben
\EE[\rho]=\frac{\Omega_d}{2}\int_{\R^{d}}\rho(x)^2\de x+\int_{\R^d}V(x)\rho(x)\de x
\een
 is an energy functional in the \emph{Thomas-Fermi} class. 
\par
\begin{theorem}[\cite{Cunden17,Cunden18}] Let $\Phi$ be  a Coulomb~\eqref{eq:Coulombdef}  or a Thomas-Fermi~\eqref{eq:ThomasFermidef} interaction. Assume that $V(x)$ is $C^{3}(\R^d)$, radially symmetric $V(x)=v(|x|)$,  with $v$ increasing and strictly convex. In the case of Coulomb interaction, assume also the growing condition  $\lim_{|x|\to\infty}\frac{V(x)}{\Phi(x)}=+\infty$. 
Then, the constrained equilibrium measure~\eqref{eq:constmin} of the energy functional~\eqref{eq:EF} for Coulomb and Thomas-Fermi gases is unique and is given by
\be
\de\rho_R(x)=
\begin{cases}
\displaystyle{\dfrac{1}{\Omega_d} \Delta V(x)\, \mathbbm{1}_{|x|\leq  R\wedge R_{\star}}\de x
+ c(R)\,\frac{\delta(|x|-R)}{\Omega_d R^{d-1}}} &\text{(Coulomb)}\\\\
\dfrac{1}{\Omega_d}\left(\mu(R)-V(x)\right) \mathbbm{1}_{|x|\leq  R\wedge R_{\star}}\de x &\text{(Thomas-Fermi)}\ .
\end{cases}
\label{eq:eqmeasureCTF}
\ee
The critical radius $\Rs$ is the unique positive solution of the equation
\be
\begin{cases}
R_\star^{d-1}v'(R_\star)=1 &\text{(Coulomb)}\\\\
v(R_{\star})\dfrac{R_{\star}^d}{d}-\displaystyle\int_0^{R_{\star}} v(r)r^{d-1}\de r=1 &\text{(Thomas-Fermi)}\ ,
\end{cases}
\label{eq:crit_R_CTF}
\ee
while $c(R)$ and $\mu(R)$ are given by
\be
\begin{cases}
c(R)=\max\{0,1-R^{d-1} v'(R)\} &\text{(Coulomb)}\\\\
\mu(R)=\max\left\{v(R_{\star}),\dfrac{d}{R^d}\left(1+\displaystyle\int_0^{R} v(r)r^{d-1}\de r\right)\right\} &\text{(Thomas-Fermi)}\ .
\end{cases}
\label{eq:CTF2}
\ee
\end{theorem}
\par
In the pulled phase, the equilibrium density of the Coulomb gas is supported on the ball of radius $R_{\star}$ and there is no accumulation of charge on the surface ($c(R)=0$ for $R\geq R_{\star}$). 
In the pushed phase, the equilibrium density in the bulk does not change, while an excess charge ($c(R)>0$ for $R<R_{\star}$) accumulates on the surface. 
\par
For the Thomas-Fermi gas,  $R_\star$ -- the edge in the pulled phase -- is determined by the condition that the gas density vanishes on the surface, i.e. $\mu(R_{\star})=v(R_{\star})$ for $R\geq R_{\star}$. In the pushed phase $R<R_{\star}$, the chemical potential increases to keep the normalisation of $\rho_R$, but there is no accumulation of charge on the surface, i.e. singular components in the equilibrium measure (otherwise the energy would diverge).
\par
A direct calculation
yields the free energy~\cite{Cunden17,Cunden18}
\be
F(R)=
\begin{cases}
\displaystyle\dfrac{1}{2} \int_{R\wedge R_{\star}}^{R_{\star}}  \dfrac{c(r)^2}{r^{d-1}}\de r&\text{(Coulomb)}\\\\
\displaystyle\dfrac{1}{2} \int_{ R\wedge R_{\star}}^{\Rs} \big(\mu(r)-v(r)\big)^2r^{d-1}\de r&\text{(Thomas-Fermi)}\ .
\label{eq:rateF_CTF}
\end{cases}
\ee 
From the exact formulae above, one can check that $F(R)$ has a jump in the third-derivative at $R=R_{\star}$.
Therefore, the critical exponent `$3$' is shared by systems with long-range (Coulomb) and zero-range (delta) interaction. This suggests that the third-order phase transition is even more universal than originally expected.
\subsubsection{Yukawa gas in generic dimension}
The ubiquity of this transition calls for a comprehensive theoretical framework, which should be valid irrespective of spatial dimension $d$ and the details of the confining potential $V$, and for the widest class of repulsive interactions $\Phi$.
\par
Fix two positive numbers $a,m>0$, and define $\Phi=\Phi_d$ as the solution of
\be
\mathrm{D}\Phi_d(x)=\Omega_d\delta(x),\quad x\in\R^d,\quad\text{where}\quad
\mathrm{D}=-a^2\Delta+m^2\ .
\label{eq:distr}
\ee
The explicit solution $\Phi_d(x)$ in terms of Bessel functions is~\eqref{eq:Yukawa_kernel}. See~Appendix~\ref{app:Yukawa}.
Note that this kernel naturally interpolates between the Coulomb electrostatic potential in free space (long-range, for $a=1$ and $m=0$), and the delta-like interaction (zero-range, for $a=0$ and $m=1$); intermediate values $a,m>0$ correspond to the Yukawa (or screened Coulomb) potential. 
\par
 \begin{assumptions}
\label{ass:potential_gen}
$V(x)$ is $C^{3}(\R^d)$, radially symmetric $V(x)=v(|x|)$,  with $v$ increasing and strictly convex. 
\end{assumptions}
\par
The condition that $v(r)$ is increasing implies the \emph{confinement} of the gas whenever $m>0$. Indeed, using the asymptotic expansion of $K_{\nu}(z)$ for large argument, the following limit holds~\cite[Eq.~10.25.3]{NIST} 
\begin{equation}
\lim_{r\to\infty}\frac{v(r)}{\varphi_d(r)}=\left(\frac{a}{m}\right)^{\frac{d-3}{2}}\lim_{r\to\infty}v(r)r^{\frac{d-1}{2}}e^{mr/a}=+\infty\ ,
\end{equation}
as long as there is screening ($m>0$). By a routine argument, this implies the existence and uniqueness of the minimiser of $\EE$ 
in $\mathcal{P}(\R^d)$. The minimiser  $\rho_{R_{\star}}$ is compactly supported $\operatorname{supp}\rho_{R_{\star}}=B_{R_{\star}}$ with $R_\star<\infty$. In particular, the support is simply connected. 
\par
The solution of the constrained equilibrium problem for the Yukawa interaction (announced in~\cite{Cunden18}) is stated below\footnote{In dimension $d=3$ the equilibrium density of Yukawa gases in a harmonic potential \emph{without} hard walls appeared in~\cite{Henning07,Levin11}.}.
\begin{theorem}
\label{thm:EM_Yukawa} 
Let $a>0$, $m\geq0$, and let $\Phi_d(x)$ be the Yukawa interaction~\eqref{eq:Yukawa_kernel} solution of~\eqref{eq:distr}.
Let $V(x)$ satisfy Assumption~\ref{ass:potential_gen}. 
In the case of Coulomb interaction, $m=0$, assume also the growing condition  $\lim_{|x|\to\infty}\frac{V(x)}{\Phi_d(x)}=+\infty$. 
Then, the constrained equilibrium measure~\eqref{eq:constmin} of the energy functional~\eqref{eq:EF} for the Yukawa gas is unique and is given by
 \be
\de\rho_R(x)=
\begin{cases}
\displaystyle\frac{1}{\Omega_d} \sigma_R(x)\mathbbm{1}_{|x|\leq R}\de x+ c(R)\,\frac{\delta(|x|-R)}{\Omega_d R^{d-1}}\ &  \text{if $R\leq R_{\star}$ (pushed phase)}\smallskip\\
\displaystyle\frac{1}{\Omega_d}\sigma_{\Rs}(x)\mathbbm{1}_{|x|\leq  R_{\star}}\de x\  & \text{if $R\geq R_{\star}$ (pulled phase)}\ ,
\end{cases} 
\label{eq:rho_Yuk}
\ee
where 
\be
\sigma_R(x)=(-a^2\Delta +m^2)(\mu(R)-V(x))
\label{eq:sigmaR}
\ee
is nonnegative for $|x|\leq R$, and $c(R)\geq0$ with $c(R)=0$ if and only if $R\ge\Rs$. 

The \emph{chemical potential} $\mu(R)$ and the \emph{excess charge} $c(R)$ are explicit functions of  $\Phi_d(x)=\varphi_d(|x|)$, $V(x)=v(|x|)$, and the constants $a$ and $m$:
 \be 
 c(R)=
a \frac{1-\left(a^2-\frac{m^2R}{d}\frac{\varphi_d(R)}{\varphi_d'(R)}\right) v'(R)R^{d-1}-\frac{m^2R^d}{d}v(R)+m^2\displaystyle\int_0^Rr^{d-1}v(r)\de r}{a-\frac{m^2R}{d}\frac{\varphi_d(R)}{a \varphi_d'(R)}}
\mathbbm{1}_{R\leq \Rs} \ ,
 \label{eq:c_Yuk}
 \ee
 \be
 \mu(R)=
  \begin{cases}
 v(R)- \frac{\varphi_d(R)}{a\varphi_d'(R)}\left( a v'(R)+\frac{c(R)}{a R^{d-1}}\right)&\text{for  $R\leq  R_{\star}$}\\\\
 v(R_{\star})- \frac{\varphi_d(R_{\star})}{\varphi_d'(R_{\star})}v'(R_{\star})&\text{for  $R\geq  R_{\star}$} \ . 
  \end{cases}
  \label{eq:mu_Yuk}
 \ee
The critical value $R_{\star}$ is the smallest positive solution of $c(R_{\star})=0$, i.e. the solution of
\be
\left(a^2-\frac{m^2R_{\star}}{d}\frac{\varphi_d(R_{\star})}{\varphi_d'(R_{\star})}\right) v'(R_{\star})R_{\star}^{d-1}+\frac{m^2R_{\star}^d}{d}v(R_{\star})-m^2\int_0^{R_{\star}}r^{d-1}v(r)\de r=1\ .
\label{eq:critR}
\ee
\end{theorem}
In particular: i) in the pulled phase the equilibrium measure is absolutely continuous with respect to the Lebesgue measure on $\R^d$; ii) when the gas is `pushed', the density in the bulk increases by a constant \emph{and} a singular component builds up on the surface of the ball $B_R$. 

\par
Theorem~\ref{thm:EM_Yukawa} implies the second main result in this paper :  \emph{the universality of the jump in the third derivative of excess free energy of a Yukawa gas with constrained volume.} This universality extends to the limit cases $m\to 0$ (Coulomb gas) and $a\to0$ (Thomas-Fermi gas) solved in~\cite{Cunden17} and~\cite{Cunden18}, respectively. See Remark~\ref{rmk:limit} below.

\begin{remark} 
\label{rmk:limit}
One can check that, in the limit cases of Coulomb ($m=0$) and delta interactions ($a\to0$), the expression~\eqref{eq:rho_Yuk}  for the equilibrium measure simplifies as in~\eqref{eq:eqmeasureCTF}. Indeed, Eq.~\eqref{eq:rho_Yuk} for $m=0$ and $a=1$ yields
\begin{equation}
\de \rho_R(x)=\frac{1}{\Omega_d} \Delta V(x)\mathbbm{1}_{|x|\leq R\wedge R_{\star}}\de x+
c(R)\,\frac{\delta(|x|-R)}{\Omega_d R^{d-1}} \ .
\end{equation}
Equation~\eqref{eq:c_Yuk} for $m=0$ becomes $c(R)=(1-v'(R)R^{d-1}) \mathbbm{1}_{R\leq \Rs}$, while  equation~\eqref{eq:critR} for $\Rs$ reduces to $v'(R_{\star})R_{\star}^{d-1}=1$, that is their respective expressions for the Coulomb gas.

The limit $a\to 0$ is more delicate. By using the asymptotic expansion of $K_{\nu}(z)\sim \sqrt{\pi/2 z}$ for large argument $z\to\infty$~\cite[Eq.~10.25.3]{NIST}, one easily gets
\begin{equation}
\frac{\varphi_d(R)}{\varphi_d'(R)} \sim -\frac{a}{m}, \quad \text{as }  a\to 0\ ,
\label{eq:fifi'}
\end{equation}
for $m>0$. Therefore, in the Thomas-Fermi gas there is no condensation, $c(R)\to0$ as $a\to0$. When $a\to 0$ and $m=1$, from equation~\eqref{eq:mu_Yuk} we recover the value of the chemical potential in the Thomas-Fermi gas~\eqref{eq:CTF2}, and
\begin{equation}
\de \rho_R(x)=\frac{1}{\Omega_d}\left(\mu(R)-V(x)\right)\mathbbm{1}_{|x|\leq R\wedge R_{\star}}\de x\ .
\end{equation}
Moreover, equation~\eqref{eq:critR} for the critical radius $\Rs$ reduces to its corresponding expression~\eqref{eq:crit_R_CTF} for the Thomas-Fermi gas.
\end{remark}

\begin{theorem} 
\label{thm:Yukawa} Under the assumptions of Theorem~\ref{thm:EM_Yukawa}, the excess free energy~\eqref{eq:def_F_gen} for a Yukawa gas is given by
\be
F(R)=\frac{1}{2} \int_{R\wedge R_{\star}}^{R_{\star}} \frac{c(r)^2}{a^2r^{d-1}}\de r\ .
\label{eq:simple_Y}
\ee
This formula implies that

\ben
F(R) >0, \quad \text{for  $R<  R_{\star}$}, \qquad  F(R) =0, \quad \text{for  $R\geq  R_{\star}$}\ ,
\label{eq:non-analyt}
\een
and
\be
F(R)\sim C_{\star}(R_{\star}-R)^3\mathbbm{1}_{R\leq R_{\star}}\ , \quad \text{as} \quad R\to \Rs\ ,
\label{eq:F_vicinity3}
\ee
with $C_{\star}>0$, that is 
\ben
F(\Rs)=F'(\Rs)=F''(\Rs)=0,\quad \text{but}\quad
\lim_{R\uparrow  R_{\star}}F'''(R) = - 3! C_{\star} <0\ .
\een

\end{theorem}

\begin{remark} 
\label{rmk:limit2}
The formula of the free energy $F(R)$ clearly matches~\eqref{eq:rateF_CTF} for $m=0$ and $a=1$. 
The limit $a\to 0$ can be obtained as above by using the limit~\eqref{eq:fifi'}. Indeed from the expression ~\eqref{eq:c_Yuk} of $c(R)$ one gets for $m=1$ and $R\leq \Rs$
\begin{equation}
\lim_{a\to 0} \frac{c(R)}{a} = \frac{1-\frac{R^d}{d}v(R)+\displaystyle\int_0^Rr^{d-1}v(r)\de r}{\frac{R}{d}}
= R^{d-1} \bigl(\mu(R) - v(R) \bigr) \ ,
\end{equation}
and from~\eqref{eq:simple_Y} we recover the expression of the free energy of the Thomas-Fermi gas~\eqref{eq:rateF_CTF}.
\end{remark}

\begin{example} Consider the two-dimensional gas $d=2$ with $a=1$, $m=0$ (Coulomb gas) in a quadratic potential $v(r)=r^2/2$. This coincides with the eigenvalue gas of the GinUE. An easy calculation from~\eqref{eq:c_Yuk} and~\eqref{eq:critR} shows that the excess charge is
\be
c(R)=1-R^2\ ,
\ee
and the critical radius is $\Rs=1$. The excess free energy can by easily computed using~\eqref{eq:simple_Y} 
\be
F_{\text{GinUE}}(R)=\frac{1}{2}\int_{R}^{1}\frac{\left(1-r^2\right)^2}{r}\de r=\frac{1}{8}(4R^2-R^4-4\log R-3) 
\ee
for $R\leq 1$, and zero otherwise, in agreement with~\eqref{eq:rateFGinUE}.
\end{example}

\subsection{Order parameter of the transition}
The universal formulae~\eqref{eq:simple_log} and~\eqref{eq:simple_Y} for the free energy $F(R)$ are remarkably simple. It feels natural to ask whether it is possible to derive them from a simpler physical argument. 
Moreover, it would be desirable to express the free energy in terms of a quantity that captures the non-analytic behaviour in the vicinity of the transition. This quantity, traditionally called \emph{order parameter} of the transition, must be zero in one phase and nonzero in the other phase. 
\par
In the following, we outline a heuristic argument that reproduces formulae~\eqref{eq:simple_log} and~\eqref{eq:simple_Y}, and identifies the `electrostatic pressure' as the order parameter of the transition.
\subsubsection{Electrostatic pressure: screened Coulomb interaction} 
For the pushed-to-pulled phase transition we can identify the order parameter as follows. Note that the derivative of $F(R)$ is essentially the variation of the free energy with respect to the volume of the system. The volume of the ball is $\operatorname{vol}(B_R)=\Omega_d R^d/d$, so that
\be
\frac{\de F}{\de R}=\left(\frac{\partial \mathrm{vol}}{\partial R}\right)\left(\frac{\partial F}{\partial \mathrm{vol}}\right)=\Omega_d R^{d-1}\left(\frac{\partial F}{\partial  \mathrm{vol}}\right)=\Omega_d R^{d-1} P\ ,
\ee
where $P=\left(\frac{\partial F}{\partial  \mathrm{vol}}\right)$ is the \emph{pressure} of the gas confined in $B_R$. 

The increase in free energy of the constrained gas should match the work $W_{R_{\star}\to R}$ done in a  compression of the gas from the initial volume $\operatorname{vol}_{i}=\operatorname{vol}(B_{R_{\star}})$ to the final volume $\operatorname{vol}_{f}=\operatorname{vol}(B_{R})$, the system being in equilibrium with density $\rho_r$  at each intermediate stage $R\leq r\leq R_{\star}$. In formulae, 
\be
F(R)=-W_{R_{\star}\to R}\quad\text{with}\quad W_{R_{\star}\to R}=\int_{\operatorname{vol}_i}^{\operatorname{vol}_f} P\, \de V =\Omega_d \int_{R_\star}^R p(r)r^{d-1}\de r\ ,
\label{eq:work}
\ee
where $P=p(r)$ is the pressure on the gas confined in $B_r$. In other words, $p(r)\Omega_d r^{d-1}\de r$ is the elementary work done on the surface of the ball of radius $r$ being compressed from $r+\de r$ to $r$. 
We now show that the `electrostatic' pressure for a Yukawa gas is quadratic in the excess charge on the surface in the pushed phase, and zero in the pulled phase.

\par
How to compute the pressure exerted by the surface's field on itself? The argument that follows is similar to the one used to evaluate the `electrostatic pressure' on a layer of charges, e.g. the surface of a charged conductor. This is a problem that some textbooks in electrostatics occasionally mention (see the classical books of Jackson~
\cite[Section 3.13]{Jackson} and Purcell~\cite[Section 1.14]{Purcell}), but which is rarely discussed due to the difficulty of making the argument rigorous for conductors of generic shapes. To stress the analogy and for the lack of a better terminology, in the following we will keep using the expressions `electrostatic' pressure, force, field, Gauss's law, etc. even though the interaction we are considering is Yukawa (screened Coulomb).  
\par
When the system is confined in a ball $B_r$ at density $\rho_r$, the pressure is given by the normal force per unit area, 
\begin{equation}
p(r) = \lim_{\Delta A\to0}\frac{|\Delta F_n|}{\Delta A}\ ,
\end{equation}
where $\Delta A=r^{d-1}\Delta\Omega$ is a small area on the sphere of radius $r$ . An intuitive guess for the force $\Delta F_n$ experienced by the surface element $\Delta A$ during the compression is 
\begin{equation}
\Delta F_n\stackrel{?}{=}(\text{charge in $\Delta A$})\times (\text{electrostatic field across $\Delta A$})\ .
\label{eq:force_guess}
\end{equation}
This is almost right, but it contains a serious flaw. Indeed, note that the electric field across $\Delta A$ includes contributions from the \emph{total} amount of charge on the sphere -- thus including the charge that is being acted upon by the field! We must therefore be 
`over-counting', as the charge on $\Delta A$  cannot act upon itself. To fix the over-counting, we may imagine to open a small hole corresponding to $\Delta A$ and compute the electrostatic field produced by the charge distribution $\rho_r$ minus the hole. A glance at Fig.~\ref{fig:layer} may be helpful. The electrostatic field in the hole is the field experienced by the charges in $\Delta A$. Hence, the force acting on $\Delta A$ is 
\begin{equation}
\Delta F_n=\text{(charge in $\Delta A$})\times (\text{electrostatic field in the hole})\ .
\label{eq:force_OK}
\end{equation}
\par
When the gas density is $\rho_r$, the amount of charge in $\Delta A$ is clearly given by 
\begin{equation}
(\text{charge in $\Delta A$})=c(r)\frac{\Delta A}{\Omega_d r^{d-1}}\ ,
\end{equation}
as the charge in the bulk is a continuous distribution and does not contribute.
\begin{figure}
\centering
  \includegraphics[width=.85\columnwidth]{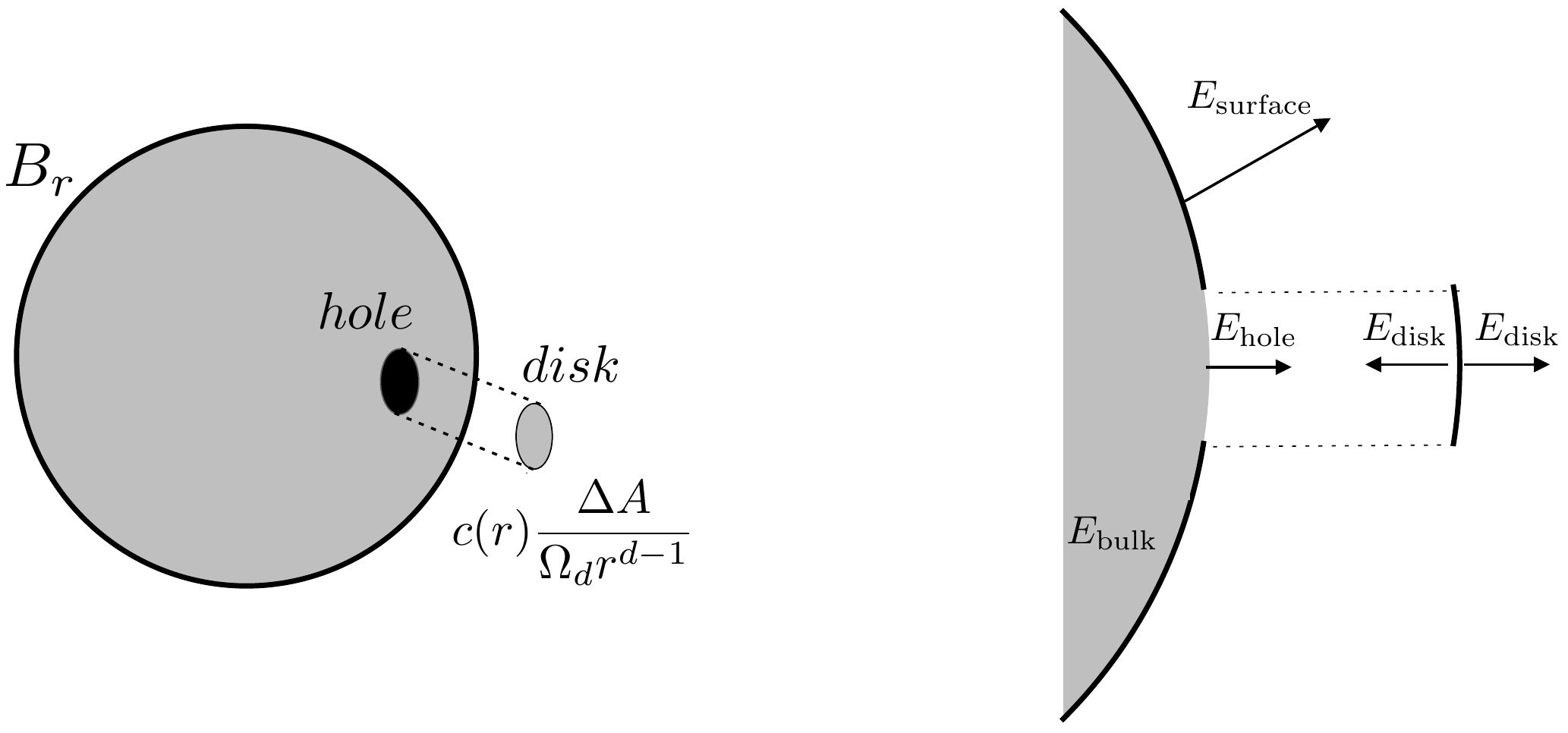}
\caption{The electrostatic pressure on the spere. Left: consider removing a small disk of area $\Delta A$ from the surface of the sphere. The electric field experienced by a point in $\Delta A$ on the sphere corresponds to the field on the hole generated by the sphere  with the small disk removed (in the limit where $\Delta A\to 0$). Right: the electric field on the hole can be computed from the total electric field generated by the sphere and the field generated by the disk using the superposition principle. Note that in the bulk the electric field must be zero at equilibrium. Therefore $E_{\text{hole}}$ must be equal in magnitude to $E_{\text{disk}}$ and directed outward.}
\label{fig:layer}       
\end{figure}
There is a clever argument to compute the electric field in the hole: we know that the field in the hole plus the field in the disk gives the field produced by the sphere 
\begin{equation*}
E_{\text{hole}}+E_{\text{disk}}=E_{\text{sphere}}=
\begin{cases}
E_{\text{surface}}, &\text{just outside the sphere}\\
E_{\text{bulk}}, &\text{just inside the sphere}\ .
\end{cases}
\end{equation*}
From basic considerations, it is clear that the field $E_{\text{bulk}}$ is zero inside the ball $B_r$  (a consequence of electrostatic equilibrium) and  perpendicular to the surface immediately outside the ball (a necessary condition for electrostatic equilibrium). Therefore $E_{\text{hole}}=-E_{\text{disk}}$. 

What is $E_{\text{disk}}$?  
Denote by $\Phi^r (x)$ the potential generated  at $x\in\R^d$ by the  charge in $\Delta A$. The electrostatic field produced by this tiny amount of charge is $-\nabla\Phi^r (x)$. Fortunately, we only need to know this in the vicinity of the disk $E_{\text{disk}}=-\nabla\Phi^r(x)$ with $|x|=r$, where the disk can be approximated by a planar surface with uniform density $c(r)\frac{\Delta A}{\Omega_d r^{d-1}}$, see Fig.~\ref{fig:layer}. Consider a small cylinder $C(h,\Delta A)$ of height $2h$ and base $\Delta A$ cutting across the surface of the ball of radius $r$ as in  Fig.~\ref{fig:Gauss}.  To compute the field on the surface, one should integrate the screened Poisson equation 
\begin{equation}
(-a^2\Delta+m^2)\Phi^r(x) =\Omega_d c(r)\frac{\Delta A}{\Omega_d r^{d-1}}\mathbbm{1}_{x\in\Delta A}\ .
\label{eq:screened_Poisson}
\end{equation}
If we integrate~\eqref{eq:screened_Poisson} over the small cylinder (see Fig.~\ref{fig:Gauss}), we get 
\begin{equation}
-a^2\int_{C(h, \Delta A)}\Delta\Phi^r(x)\de x+m^2\int_{C(h,\Delta A)}\Phi^r(x)\de x=c(r)\frac{\Delta A}{r^{d-1}}\ .
\end{equation}
By symmetry, the electrostatic field $ -\nabla\Phi^r$ is perpendicular to the bases, along $\hat{x}=x/r$, and, using the divergence theorem,
\begin{equation}
-a^2 \Delta A \left[\nabla\Phi^r(x+h \hat{x} )\cdot \hat{x} - \nabla\Phi^r(x-h \hat{x} )\cdot \hat{x}  \right]+m^2\int_{C(h,\Delta A)}\Phi^r(x)\de x=c(r)\frac{\Delta A}{r^{d-1}}\ .
\end{equation}
Letting $h\to0$,  the volume integral on the left hand side vanishes and we get
\begin{equation}
E_{\text{disk}}=-\nabla\Phi^r(x)=
\begin{cases}
+\dfrac{1}{2}\dfrac{c(r)}{a^2r^{d-1}} \hat{x} &\text{immediately outside the disk}\\\\
-\dfrac{1}{2}\dfrac{c(r)}{a^2r^{d-1}} \hat{x} &\text{immediately inside the disk}\ .
\end{cases}
\end{equation}
Therefore, we conclude that
\be
E_{\text{hole}}=\frac{1}{2}\dfrac{c(r)}{a^2r^d}x\ .
\label{eq:E_hole}
\ee
As a byproduct, we see instead that $E_{\text{surface}}=\frac{c(r)}{a^2r^d}x$; this is the familiar statement that, at equilibrium, the electrostatic field generated by a charged conductor immediately outside is perpendicular to its surface and proportional to the charge density.
\par
 Putting everything together, we obtain 
\begin{equation}
p(r) = \lim_{\Delta A\to0}\frac{1}{\Delta A}\left(c(r)\frac{\Delta A}{\Omega_d r^{d-1}}\times\frac{1}{2}\dfrac{c(r)}{a^2r^{d-1}}\right)=\frac{1}{2}\frac{1}{\Omega_d r^{d-1}}\frac{c^2(r)}{a^2 r^{d-1}}\ .
\label{eq:pressure_res}
\end{equation}
Plugging~\eqref{eq:pressure_res} into~\eqref{eq:work} we get
\be
W_{R_{\star}\to R}=\Omega_d \int_{R_\star}^R p(r)r^{d-1}\de r=- \frac{1}{2} \int_{R}^{R_{\star}}\frac{c^2(r)}{a^2 r^{d-1}}\de r\ ,
\ee  
which is exactly (minus) the excess free energy~\eqref{eq:simple_Y}.
\begin{figure}
\centering
  \includegraphics[width=.55\columnwidth]{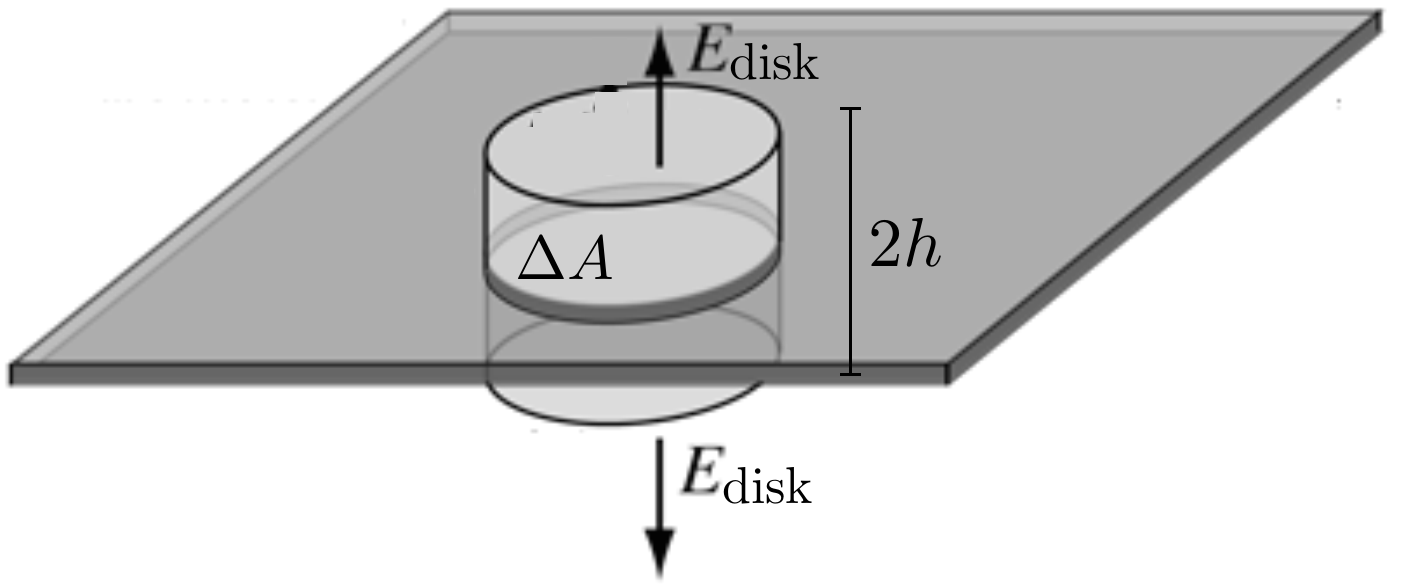}
\caption{The field in the vicinity of the sphere can be computed integrating the screened Poisson equation over a small volume enclosed by a  `Gauss surface'. (At that scale, the sphere can be approximated by an infinite plane.)}
\label{fig:Gauss}       
\end{figure}
\subsubsection{Electrostatic pressure: random matrices} 
\label{subsec:log}
The argument outlined above can be repeated almost verbatim for the log-gas on the line (eigenvalues of random matrices). However there is a twist in the computation. Again, by conservation of energy, the increase in free energy must match  the work $W_{R_{\star}\to R}$ done in a  compression of the gas from the initial volume (length) $\mathrm{vol}_{i}=\operatorname{vol}(B_{R_{\star}})=2 R_{\star}$ to the final volume $\mathrm{vol}_{f}=\operatorname{vol}(B_{R})=2R$, with the system in equilibrium with density $\rho_r$  at each intermediate stage $R\leq r\leq R_{\star}$. In formulae, 
\be
F(R)=-W_{R_{\star}\to R}\quad\text{with}\quad W_{R_{\star}\to R}=\int_{\mathrm{vol}_i}^{\mathrm{vol}_f} P\,\de V =2 \int_{R_\star}^R p(r)\de r\ ,
\label{eq:work2}
\ee
where $2p(r)\de r$ is the elementary work done in an infinitesimal compression (the factor $2$ comes from  axial symmetry). 
\par
When the system is confined in a ball $B_r$ at density $\rho_r$, the pressure is given by the normal force per unit length. The force $F_n(x)$ at point $x$ is equal to the charge $\rho_r(x)\de x$ in the infinitesimal segment $\de x$ around $x$ times the electric field. To proceed in the computation it is convenient to use complex coordinates (recall that $-\log|x|$ is the Coulomb interaction in dimension $d=2$).

The electric field generated by $\rho_r(y)\de y$ at $z\in\C\setminus[-r,r]$ is the Stieltjes transform 
\ben
G(z)=\int\frac{\rho_r(y)}{z-y}\de y\ .
\een
Note that $\rho_r(z)=0$ when $z\notin[-r,r]$.
Using Gauss's theorem (Plemelj formula), when $z$ approaches the real axis, the field generated by $\rho_r$ is 
\be
\lim_{\epsilon\downarrow 0}G(x\pm \mathrm{i}\epsilon)= \fint \frac{\rho_r(y)}{x-y}\de y\mp \mathrm{i} \pi\rho_r(x) \ = V'(x)\mp \mathrm{i} \pi\rho_r(x)\ ,
\ee
when $|x|<r$, where $\fint$ denotes Cauchy's principal value. Note that, at equilibrium, the real part of $G(x\pm \mathrm{i}\epsilon)$ cancels with the field $-V'(x)$ generated by the external potential $V(x)$ (the net tangential field must be zero).
Therefore, the electric field experienced by a point in the vicinity of $x$ is 
\begin{align}
\operatorname{Re} E&= \lim_{\epsilon\downarrow 0}\operatorname{Re} G(x\pm \mathrm{i}\epsilon)-V'(x)=0\\
\operatorname{Im} E&= \lim_{\epsilon\downarrow 0}\operatorname{Im} G(x\pm \mathrm{i}\epsilon)=\mp  \pi\rho_r(x)\ .
\end{align}

The force (\text{$=$electric field$\times$ charge}) per unit length is
\ben
\lim_{\Delta \ell\to0}\frac{F_n(z)}{\Delta \ell}=E(z)\rho_r(z)\ .
\een
In the vicinity of $x$ it becomes
\be
\lim_{\epsilon\downarrow 0}\lim_{\Delta \ell\to0}\frac{F_n(x\pm \mathrm{i}\epsilon)}{\Delta \ell}=\mp\mathrm{i} \pi\rho_r^2(x)=\mp \frac{\mathrm{i}}{\pi}\frac{P_r^2(x)}{r^2-x^2}\ .
\label{eq:force_log-gas}
\ee
Note, in particular, that the tangential force experienced by a point $x$ inside the conductor is zero (as it should be at equilibrium). 

At the edge $x=r$ (similar considerations for $x=-r$) the situation is more delicate. By symmetry, the total field $E_{\text{edge}}$ generated by $\rho_r(x)$ at $x=r$ must be directed along the $x$-axis, but must be zero for $x<r-\epsilon$. Therefore, repeating the argument of the previous section, the field experienced by the `hole' at the edge $x=r$ is half the field generated by $\rho_r(x)$ at $r+\epsilon$, i.e. $E_{\text{hole}}=(1/2)E_{\text{edge}}$.

To compute the pressure, we look at the edge $x=r$, we sum the forces on a small circular contour of radius $\epsilon$ centred at $ r$, and then we take the limit $\epsilon\to0$ (see Fig.~\ref{fig:contour})
\begin{align}
p(r)=\lim_{\Delta \ell\to0}\frac{|F_n(r)|}{\Delta \ell}=\left|\lim_{\epsilon\to0}\int_{c_{\epsilon}}\frac{1}{2}\frac{\mathrm{i}}{\pi}\frac{P_r^2(z)}{ r^2-z^2}\de z\right| ,
\end{align}
where the factor $1/2$ comes from the fact that $E_{\text{hole}}=(1/2)E_{\text{edge}}$.
Remembering that the force is eventually zero on the semicircular part of the contour in the bulk,  the integral is given by $\pi \mathrm{i}$ times the residue at $z=r$:
\be
p(r)=\left|\pi \mathrm{i}\operatorname{Res}_{z=r}\frac{\mathrm{i} \, P_r^2(z)}{2\pi (r^2-z^2)}\right|=\frac{P_r^2(r)}{4r}\ .
\label{eq:respp}
\ee
Inserting this formula in~\eqref{eq:work2} we indeed recover~\eqref{eq:simple_log}.
\begin{remark}
\label{rmk:nonsym}
These electrostatic considerations indicate a route to compute large deviation functions for extreme eigenvalues of random matrices more general that those fulfilling Assumption~\ref{ass:potential_log-gas}. For instance, one may ask whether it is possible to obtain an electrostatic formula for the large deviation of the \emph{top eigenvalue} $x_{\max}$  (i.e. the rightmost particle) of a random matrix from a $\beta$-ensemble. While this question does not fit into the symmetric setting considered so far, it can be nevertheless easily answered within the electrostatic framework developed above. For a one-cut matrix model with typical top eigenvalue equal to $b_{\star}$, the rate function function $F(b)$ in the large deviation decay
\be
\operatorname{Pr}\left(x_{\max}\leq b \right)\approx e^{-\beta N^2 F(b)},
\ee
is zero in the `pulled phase' ($F(b)=0$ for $b\geq b_{\star}$) and nonnegative on the left  $b\leq b_{\star}$. One can reproduce the previous heuristic considerations and argue that the rate function is the work done in pushing the gas with a hard wall from $b_{\star}$ to $b$
\be
F(b)=-\int_{b_{\star}}^{b }p(u)\de u\ ,
\label{eq:rateF_RMT}
\ee
where the electrostatic `pressure' experienced by the gas is now given by
\be
p(b)=\frac{\pi^2}{2}|\operatorname{Res}_{z=b}\rho_b^2(z)|\ .
\label{eq:respp2}
\ee
$\rho_b(x)$ is the density of the log-gas (with support in $[a,b]$) constrained to stay on the left of the hard wall at $x=b$.
\par
One can easily check the validity of this formula for a few cases already considered in previous works. For instance, the constrained density of the GUE ensemble with one hard wall is~\cite{Dean06,Dean08}
\be
\rho_b(x)=\frac{1}{2\pi}\sqrt{\frac{x-a}{b-x}}(b-a-2x),\quad\text{with $a=-\frac{2\sqrt{b^2+6}-b}{3}$}\ .
\ee
In this case 
\be
p(u)=\frac{\pi^2}{2}|\operatorname{Res}_{z=u}\rho_u^2(z)|=\frac{1}{27} \left(u^3+\sqrt{u^2+6} u^2+6 \sqrt{u^2+6}-18 u\right),
\ee
and inserting this expression into~\eqref{eq:rateF_RMT} one recovers the known result~\cite[Eq. (21)]{Majumdar14}. 
\par
The calculations are similar for random matrices from the Wishart ensemble (a log-gas on the positive half-line). For simplicity we report the calculation for matrices with $c=1$ (see~\cite[Sec. 3.1]{VMB07}) where the equilibrium measure is the Marchenko-Pastur distribution supported on the interval $[0,4]$. The constrained density with a wall at $b\leq4$ reads 
\be
\rho_b(x)=\frac{1}{2\pi}\frac{b/2+2-x}{\sqrt{x(b-x)}}\ .
\ee
Computing the residue
\be
p(u)=\frac{\pi^2}{2}|\operatorname{Res}_{z=u}\rho_u^2(z)|=\frac{u^2-8 u+16}{32 u}\ .
\ee
The rate function for $b\leq 4$ is
\be
F(b)=-\int_{4}^{b}p(u)\de u=-\frac{b^2}{64}+\frac{b}{4}-\frac{1}{2} \log \frac{b}{4}-\frac{3}{4}\ ,
\ee
which coincides with~\cite[Eq. (35)-(36)]{VMB07}. 
\end{remark}
\begin{remark}
The calculation of the work done in a compression of a log-gas in dimension $d=1$ bears a strong resemblance to a way to calculate the work (energy) per unit fracture length for a crack propagating in a continuous medium. In linear elastic fracture mechanics, Cherepanov~\cite{Cherepanov67} and Rice~\cite{Rice68} have independently developed a line integral called the $J$-integral that is contour independent. The usefulness of this integral comes about when the contour encloses the crack-tip region, as this is where the most intense (actually divergent) fields are found (c.f. the edge of the cut in the log-gas). Evaluating the $J$-integral then gives the variation of elastic energy. An analogous integral in electrostatics has been discovered later~\cite{Garboczi88}. 

It is likely that the computations outlined above can be recast in the language of linear elastic mechanics/electrostatics. The  link between `eigenvalues of random matrices' and  the `theory of fractures' suggested here will be explored in future works.

\end{remark}
\begin{figure}
\centering
  \includegraphics[width=1\columnwidth]{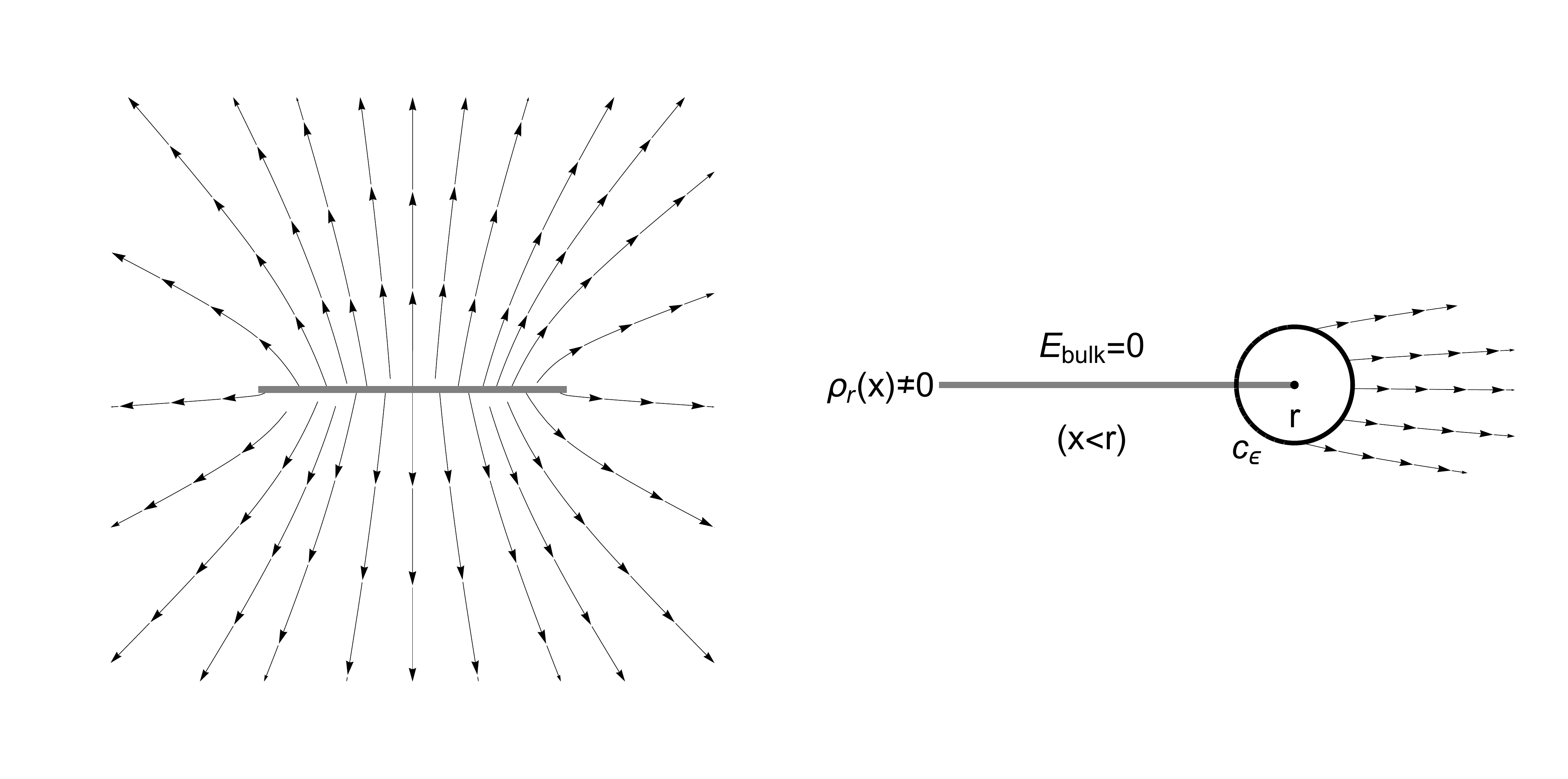}
\caption{Left: Electrostatic field generated by the pushed log-gas $\rho_r(x)$ (GUE in this example). Right: Contour of integration $c_{\epsilon}$ around the edge $x=r$. The field on the contour is not zero. As $\epsilon\to0$ only the right semicircular arc contributes to the integral. } 
\label{fig:contour}       
\end{figure}

\subsection{Outline of the paper.} The rest of the article is organised as follows. In Section~\ref{sec:equilibrium} we recall some general variational arguments for the solution of the constrained equilibrium problem. 
Section~\ref{sec:log-gas} contains the proof of Theorem~\ref{thm:EM_log-gas}  and Theorem~\ref{thm:log-gas2} for the log-gas. In Section~\ref{sec:Yukawa} we present the proof of Theorem~\ref{thm:EM_Yukawa} (equilibrium problem for Yukawa interaction) and Theorem~\ref{thm:Yukawa}  (universality of the jump in the third derivative of the free energy). Finally, the Appendices~\ref{app:Yukawa} and~\ref{app:lemmas} contain the proof of some technical lemmas.

\section{Variational approach to the constrained equilibrium problem}
\label{sec:equilibrium}
We resort to a variational argument to derive necessary and sufficient conditions for $\rho_R$ to be the minimiser of the energy functional $\EE$ over $\mathcal{P}(B_R)$.  (These arguments are not new at all. They appear in many different forms and specialisations in the literature, see, e.g.,~\cite{Bernoff11,Deift99,Levy15}.) 
\par
Denote by $\rho_{R}\in\mathcal{P}(B_R)$ a local equilibrium and let
\be
\rho=\rho_{R}+\sigma\in\mathcal{P}(B_R)\ .
\ee
Here $\sigma$ is a (small) perturbation of zero mass. Of course, the perturbation $\sigma$ must be nonnegative on $(\operatorname{supp}\rho_R)^c$, the complement of $\operatorname{supp}\rho_R$. 
\par
The functional $\EE$ is quadratic in $\rho$, hence
\be
\EE[\rho]=\EE[\rho_{R}]+\EE_1[\rho_{R},\sigma]+\EE_2[\sigma,\sigma]\ ,
\ee
where $\EE_1$ and $\EE_2$ are the first and second variations, respectively. They are given explicitly by
\begin{align}
\EE_1[\rho_{R},\sigma]&=\int\left(\int\Phi(x-y)\de\rho_{R}(y)+V(x)\right)\de\sigma(x)\ ,
\label{eq:1st_var}\\
\EE_2[\sigma,\sigma]&=\frac{1}{2}\iint\Phi(x-y)\de\sigma(y)\de\sigma(x)\ .
\label{eq:2nd_var}
\end{align}
A sufficient condition for $\rho_{R}$ to be the global minimiser in $\mathcal{P}(B_R)$ is that $\EE_1[\rho_{R},\sigma]\geq0$ and $\EE_2[\sigma,\sigma]>0$ for all perturbations $\sigma$. 
\par
Consider the first variation~\eqref{eq:1st_var} and suppose that $\sigma$ varies among the perturbations whose support lies in $\operatorname{supp}\rho_R$. Because $\sigma$ is arbitrary and zero mass, for the first variation $\EE_1$ to vanish, it must be true that
\be
\int\Phi(x-y)\de\rho_{R}(y)+V(x)=\mu(R),\quad \text{for $x\in\operatorname{supp}\rho_R$}\ ,
\label{eq:E-L1}
\ee
for some constant $\mu(R)$. Consider now perturbations $\sigma$  with support in $B_R$. Remembering that $\sigma\geq0$ in $(\operatorname{supp}\rho_R)^c$, we see that a sufficient condition for $\EE_1\geq0$ is that 
\be
\int\Phi(x-y)\de\rho_{R}(y)+V(x)\geq\mu(R),\quad \text{for $x\in B_R\setminus\operatorname{supp}\rho_R$}\ .
\label{eq:E-L2}
\ee
The conditions~\eqref{eq:E-L1}-\eqref{eq:E-L2} are known as Euler-Lagrange (E-L) conditions and can be summarised by saying that if $\rho_R$ is a minimiser of $\EE$ in $\mathcal{P}(B_R)$ then there exists a constant $\mu(R)\in\R$ such that
\be
\begin{cases}
\displaystyle\int\Phi(x-y)\de\rho_R(y)=\mu(R)-V(x)&\text{in $\operatorname{supp}\rho_R$}\ ,\\
\displaystyle\int\Phi(x-y)\de\rho_R(y)\geq\mu(R)-V(x)&\text{in $B_R$}\ .
\end{cases}
\label{eq:E-L}
\ee
The constant $\mu(R)$ is called \emph{chemical potential}. 
In general, the E-L conditions provide the saddle-point(s) of the energy functional. To show that $\rho_R$ is actually the minimiser it remains to check that the $\EE$ is strictly convex, i.e. that the second variation $\EE_2>0$. Denote by $\widehat{\sigma}$ the Fourier transform of $\sigma$. Then,
\begin{align*}
\EE_2[\sigma,\sigma]&=\frac{1}{2}\iint\Phi(x-y)\de\sigma(y)\de\sigma(x)\\
&=\frac{1}{2}\iint\widehat{\Phi}(k)|\widehat{\sigma}(k)|^2\de k\ .
\end{align*}
For the class of interactions considered in this paper,  $\widehat{\Phi}(k)>0$ (see Appendix~\ref{app:Yukawa}), and this ensures that $\EE_2[\sigma,\sigma]>0$.
Therefore, for all $R>0$, the solution of the E-L conditions is the unique minimiser of the energy functional $\EE$ in $\mathcal{P}(B_R)$.

\section{Proof of Theorems~\ref{thm:EM_log-gas} and~\ref{thm:log-gas2} }
\label{sec:log-gas}
Part i) of Theorem~\ref{thm:EM_log-gas} is a classical result in potential theory, see~\cite{Deift99}. There are two cases:
\begin{itemize}
\item[-] if the walls are not active (pulled phase) the density is supported on $\operatorname{supp}\rho_{\Rs}=[-\Rs,\Rs]$ with $\Rs$ solution of
\be
\frac{1}{\pi}\int_{-\Rs}^{+\Rs}\frac{V'(x)}{\sqrt{\Rs^2-x^2}}\de x=1\ ,
\label{eq:edge_log-gas}
\ee
and the density is given by Tricomi's formula~\cite{Tricomi}
\be
\rho_{\Rs}(x)=\frac{1}{\pi\sqrt{\Rs^2-x^2}}\left(1-\fint_{-\Rs}^{+\Rs}\frac{1}{\pi}\frac{\sqrt{\Rs^2-t^2}V'(t)}{x-t}\de t\right)\ ,
\label{eq:Tricomi}
\ee
where $\fint$ denotes Cauchy's principal value.
\item[-] if the walls are active (pushed phase) the density is supported on $\operatorname{supp}\rho_{R}=[-R,R]$ and is given by~\eqref{eq:Tricomi} with the replacement $\Rs\mapsto R$.
\end{itemize}
The new content of the theorem is part ii). To prove it, we expand the potential $V$ and the regular part of the density into Chebyschev polynomials

\ben
V(Ru)=\sum_{n\geq 0}c_n(R)T_n(u)\ ,\qquad
P_R(Ru)=\sum_{n\geq 0}a_n(R)T_n(u)\ ,
\een
where
\ben
a_n(R)=\frac{1}{h_n}\int_{-1}^1  \frac{P_R(Ru)T_n(u)}{\sqrt{1-u^2}}\de u\ ,\qquad 
c_n(R)=\frac{1}{h_n}\int_{-1}^1  \frac{V(Ru)T_n(u)}{\sqrt{1-u^2}}\de u\ .
\een
A priori, the above expansions are in $L^2([-1,1])$. 
In fact, $V\in C^3$ implies that $c_n(R)=\mathrm{O}(n^{-3})$ so that the series $\sum_{n\geq 0}c_n(R)T_n(u)$ and its derivative are pointwise convergent almost everywhere to $V$ and $V'$, respectively.
We will see in the course of the proof that the  absolute convergence of $\sum_n nc_n(R)$ implies the pointwise convergence of $\sum_na_n(R)T_n(u)$, too. Note also that $c_n=0$ if $n$ is odd (the potential $V(x)$ is symmetric by assumptions). To proceed we use the following identity.
\begin{lemma}
\label{lem:Chebyshevprime}
Let $n\geq0$ be an even integer. Then,
 \begin{equation}
 uT_{n}'(u)=nT_0(u)+nT_n(u)+2n\left(T_2(u)+T_4(u)+\cdots+T_{n-2}(u)\right)\ .
 \label{eq:PPformula}
\end{equation}
\end{lemma}
We first express the equation~\eqref{eq:edge_log-gas} for the critical radius $\Rs$ in terms of the $c_n$'s. After the change of variable $x=\Rs u$, \eqref{eq:edge_log-gas} becomes
\begin{align*}
1=\frac{\Rs}{\pi}\int_{-1}^1uV'(\Rs u)\frac{\de u}{\sqrt{1-u^2}}=\frac{1}{\pi}\int_{-1}^1\sum_{n\geq0}c_n(\Rs)uT'_n(u)\frac{\de u}{\sqrt{1-u^2}}=\sum_{n\geq0}nc_n(\Rs)\ ,
\end{align*}
where we used Lemma~\ref{lem:Chebyshevprime} and the orthogonality relation~\eqref{eq:Cheby_orth}. Note that $nc_n(\Rs)=\mathrm{O}(n^{-2})$ and hence the series is absolutely convergent. This proves that $\Rs$ is the solution of~\eqref{eq:radius_log}.
\par
The Chebyshev polynomials satisfy the following electrostatic formula (for a proof see Appendix~\ref{app:lemmas}).
\begin{lemma}[Chebyshev electrostatic formula]
\label{lem:electrostatic_log-gas} Let $x\in\R$. Then
\be
-\int_{-1}^1\log |x-y|\frac{T_n(y)}{\pi\sqrt{1-y^2}}\de y=
\begin{cases}
\delta_{n,0}\log2+(1-\delta_{n,0})\dfrac{1}{n}T_n(x)&|x|\leq 1\\\\
\dfrac{1}{n}e^{-nz}\quad (\text{with $x=\cosh z$})&|x|\geq 1\ .
\end{cases}
\label{eq:electrostatic_log-gas}
\ee

\end{lemma}
From the E-L equation, an application of the Chebyshev electrostatic formula~\eqref{eq:electrostatic_log-gas} gives 
\begin{align*}
-&\int_{-R}^{R}\log|x-y|\rho_R(y)\de y+V(x)\\
&=-\log\frac{R}{2}+a_0(R)c_0(R)+\sum_{n\geq 1}\left(\frac{1}{n}a_n(R)+c_n(R)\right)T_n\left(\frac{x}{R}\right)=\mu(R)\quad \text{if $|x|\leq R$}\ .
\end{align*}
This equation  and the normalisation of $\rho_R(x)$ imply that
\begin{equation}
a_0(R)=1,\qquad \frac{1}{n}a_n(R)+c_n(R)=0\quad\forall n\geq 1\ . 
\label{eq:a_c}
\end{equation}
In particular, we have an explicit formula for the chemical potential
\begin{equation}
\mu(R)=-\log\frac{R}{2}+c_0(R)=-\log\frac{R}{2}+\int_{-R}^R \frac{V(x)}{\pi\sqrt{R^2-x^2}}\de x\ .
\label{eq:mu_log-gas}
\end{equation}
Equation~\eqref{eq:a_c} shows that 
\be
P_R(Ru)=1-\sum_{n\geq1}nc_n(R)T_n(u)\ .
\ee
The sequence $nc_n(R)$ is $\mathrm{O}(n^{-2})$ and hence the series is pointwise convergent almost everywhere. This concludes the proof of Theorem~\ref{thm:EM_log-gas}.

\par 
To prove Theorem~\ref{thm:log-gas2} we begin by computing $F(R)$:
\begin{align}
\nonumber F(R) &=\frac{1}{2}\left(\mu(R)+\int_{-R}^R V(x)\rho_R(x)\de x\right)\nonumber\\
\nonumber &=\frac{1}{2}\left(\mu(R)+\int_{-1}^1 V(Ru)\frac{P_R(Ru)}{\pi\sqrt{1-u^2}}\de u\right)\nonumber\\
\nonumber &=\frac{1}{2}\left(\mu(R)+\sum_{n,m\geq 0}\int_{-1}^1 c_n(R) a_m(R)\frac{T_n(u)T_m(u)}{\pi\sqrt{1-u^2}}\de u\right)\nonumber\\
\nonumber &=\frac{1}{2}\left(\mu(R)+ c_0(R)a_0(R)+\frac{1}{2}\sum_{n\geq 1}c_n(R)a_n(R)\right)\nonumber\\
&= \frac{1}{2}\left(-\log\frac{R}{2}+2c_0(R)-\sum_{n\geq 1}\frac{n c_n^2(R)}{2}\right)\ .
\end{align}
(The first identity follows from the definition of $F(R)$ as energy  difference~\eqref{eq:EF_log}, and the E-L condition; then we expanded in Chebyshev polynomials and used their orthogonality relation; the last equality follows from~\eqref{eq:mu_log-gas} and~\eqref{eq:a_c}.) 
Therefore
\begin{align}
F'(R)&=-\frac{1}{2R}\left(1-2 Rc_0'(R)+R\sum_{n\geq 1}nc_n(R)c_n'(R)\right)\ .
\label{eq:F_formula}
\end{align}
We want to prove that the above expression is equal to $-P_R(R)^2/(2R)$. First, notice that 
$P_R(R)$ is 
$$
P_R(R)=\sum_{n\geq0}a_n(R)T_n(1)=\sum_{n\geq0}a_n(R)=1-\sum_{n\geq1}nc_n(R)\ ,
$$
so that 
\begin{equation}
-\frac{P_R(R)^2}{2R}=
-\frac{1}{2R}\left(1-\sum_{n\geq 1}nc_n(R)\right)^2\ .
\end{equation}
Comparing with~\eqref{eq:F_formula}, the identity to show to complete the proof is
\begin{equation}
1-2 Rc_0'(R)+R\sum_{n\geq 1}nc_n(R)c_n'(R)\stackrel{?}{=}\left(1-\sum_{n\geq 1}nc_n(R)\right)^2.
\label{eq:identity}
\end{equation}
Using Lemma~\eqref{lem:Chebyshevprime} and the identity $u\partial V(Ru)/\partial u=R\partial V(Ru)/\partial R$, we have
\ben
 c'_n(R)=
 \begin{cases}
\displaystyle \sum_{m\geq1} \dfrac{mc_m(R)}{R}&\text{if $n=0$}\\
\displaystyle\sum_{m\geq n} \dfrac{2mc_m(R)}{R}-\frac{nc_n(R)}{R}&\text{if $n>0$}\ .\\
 \end{cases}
\een
 Therefore we find
 \begin{align*}
 &1-2 Rc_0'(R)+R\sum_{n\geq 1}nc_n(R)c_n'(R)\\
 &=1-2 \sum_{m\geq1} mc_m(R)+\sum_{n\geq 1}nc_n(R)\left(\sum_{m\geq n} 2mc_m(R)-nc_n(R)\right)\\
 &=1-2 \sum_{m\geq1} mc_m(R)+\sum_{m,n\geq 1}nc_n(R)mc_m(R)\\
 &=\left(1-\sum_{n\geq1} nc_n(R)\right)^2\ .
 \end{align*}
 This concludes the proof of the integral formula~\eqref{eq:simple_log}. 
 \par 
 We proceed to prove that $F(R)$ has a jump in the third derivative. First, note that $F(R)$ is identically zero for $R\geq \Rs$, while $F(R)\geq0$ for $R\leq \Rs$. From~\eqref{eq:simple_log} and the fact that $P_{\Rs}(\Rs)=0$, one sees that
 \begin{align}
\lim_{R\uparrow  R_{\star}}F(R)&=\frac{1}{2} \int_{R_{\star}}^{R_{\star}} \frac{P_r(r)^2}{r}\de r=0\ ,\\
\lim_{R\uparrow  R_{\star}}F'(R)&=-\frac{1}{2}\frac{P_{\Rs}(\Rs)^2}{R_{\star}}=0\ ,\\
\lim_{R\uparrow  R_{\star}}F''(R)&=-\frac{2P_{\Rs}(\Rs)P_{\Rs}'(\Rs)R_{\star}-P_{\Rs}(\Rs)^2}{2R_{\star}^{2}}=0\ .
\end{align}
On the other hand, 
\begin{align}
\lim_{R\uparrow  R_{\star}}F'''(R)&=-\frac{P_{\Rs}'(R_{\star})^2}{R_{\star}}<0\ .
\label{eq:3rd_log-gas_der}
\end{align}
Indeed, by Assumption~\ref{ass:potential_gen} on the potential $V(x)$, it is easy to check that $P_{\Rs}'(\Rs)<0$ strictly. \qed.

\section{Proof of Theorems~\ref{thm:EM_Yukawa} and~\ref{thm:Yukawa} }
\label{sec:Yukawa}
A systematic analysis of the equilibrium problem for the screened Coulomb interaction does not seem to have appeared in the existing literature. Here we solve the problem (Theorem~\ref{thm:EM_Yukawa}). In this Section we put forward a sensible ansatz for $\rho_R$  depending on two parameters (chemical potential $\mu$ and surface charge $c$), that we then prove to be the minimiser by imposing the E-L conditions that fix $\mu=\mu(R)$ and $c=c(R)$. The only thing that needs to be checked is the positivity of the candidate solution  (Remark~\ref{rmk:positivity}). In solving the problem, we will make the most of its spherical symmetry. A key technical ingredient in the derivation will be a shell integration lemma (Lemma~\ref{lem:shell_thm_Yukawa}), the analogue of Newton's formula~\eqref{eq:Newton} for a Yukawa interaction in generic dimension $d\geq1$. 
\subsection{General form of the constrained minimisers}
The usual strategy in these minimisation problems is to look for a candidate solution of the E-L conditions~\eqref{eq:E-L}. 
The condition $\EE_2>0$ guarantees that the saddle-point is the minimiser in $\mathcal{P}(B_R)$. 
\par
In absence of volume constraints, i.e. $B_R=\R^d$, the global minimiser is supported on a ball of radius $R_{\star}$. For $R>R_{\star}$, the density of the gas does not feel the hard walls and $\rho_R=\rho_{R_{\star}}$. We anticipate here that for $R<R_{\star}$, $\operatorname{supp}\rho_R=B_R$ (thus explaining the name 'pushed phase') so that the second E-L condition~\eqref{eq:E-L} is immaterial.
A first idea to find $\rho_R$ is to use the property $\mathrm{D}\Phi_d(x)=\Omega_d\delta(x)$. Applying $\mathrm{D}$ to both sides of the first E-L condition we \emph{formally} get
\be
\de\rho_R(x)\stackrel{?}{=}\frac{1}{\Omega_d}\mathrm{D}\bigl(\mu(R)-V(x)\bigr)\mathbbm{1}_{|x|\in R\wedge R_{\star}}\de x\ .
\label{eq:ansatz1}
\ee
However the above ansatz is, in general, incorrect. In particular, the equality~\eqref{eq:ansatz1} is not true if $\rho_R$ contains singular components.
\par
One can prove the absence, at equilibrium, of condensation of particles in the bulk (absence of $\delta$-components). To see this, one compares the energy of a density containing a $\delta$-function in the bulk to one where the $\delta$-function has been replaced by a narrow, symmetric mollification  $\delta_{\epsilon}$ (see~\cite[Section 3.2.1]{Bernoff11} for details). This argument fails on the boundary where it is not possible to consider symmetric mollifications $\delta_{\epsilon}$ contained in the support.

As a matter of fact, condensation  of particles, though not possible in the bulk, can and do occur on the boundary of the support! For Coulomb gases ($\mathrm{D}=-\Delta$) this statement is the well-known fact that, at electrostatic equilibrium, any excess of charge must be distributed on the surface of a conductor. In a rotationally symmetric problem, any accumulation of charge must be uniformly distributed on the `surface', i.e. the boundary of the support of $\rho_R$. Using the same argument, one can argue that in the unconstrained problem -- pulled phase -- the accumulation of charge on the surface is not possible (otherwise it would be possible to mollify the $\delta$-components on the surface, lowering the energy).
Therefore, a more appropriate ansatz is
 \be
\de\rho_R(x)=\frac{1}{\Omega_d} \mathrm{D}\bigl(\mu(R)-V(x)\bigr)\mathbbm{1}_{|x|\in R\wedge R_{\star}}\de x
+ c(R)\,\frac{\delta(|x|-R)}{\Omega_d R^{d-1}}
\ ,
\label{eq:ansatz2}
\ee
for some constants $\mu(R)$ and $c(R)$ to be determined. In the following we will show that, for all $R$, there exists a unique choice of $\mu(R)$ and $c(R)$ such that~\eqref{eq:ansatz2} is a saddle-point.

\begin{remark} 
\label{rmk:positivity}
In spherical coordinates $x=r \omega$, $r = |x|\geq0$, $\omega = x/|x|\in S^{d-1}$, Eq.~\eqref{eq:ansatz2}
 reads
\begin{align}
\de\rho_R(r,\omega)=\frac{1}{\Omega_d}\Bigl[&f(r)\mathbbm{1}_{r\leq R\wedge R_{\star}}+\frac{c(R)}{R^{d-1}}\delta(r-R)\Bigr]\de r\de\omega\ ,
\label{eq:rhoR_Yukawa}
\end{align}
where we split the contribution on the surface $c(R)$ and the density $f(r)$ in the bulk
\be
f(r)=m^2(\mu(R)-v(r))r^{d-1}+a^2(r^{d-1}v'(r))'\ .
\label{eq:bulk}
\ee
\par
The equilibrium measure $\rho_R$ is continuous in the bulk $|x|<R$, and contains, in general, a singular component on the surface $|x|=R$.  
From the explicit formulae~\eqref{eq:c_Yuk}-\eqref{eq:mu_Yuk}  for $c(R)$ and $\mu(R)$, it is not yet obvious to see that $\rho_R$ is a positive measure, nor that the critical radius $R_{\star}$ is nonzero. 
We show here that the equilibrium measure is both positive in the bulk and on the surface. 
\par
Consider first the pushed phase $R<R_{\star}$. Starting from the surface, note that $c(R)$ is continuous for $R\geq0$ and differentiable at least twice for $R>0$ and $R\neq R_{\star}$. Moreover,
\begin{align}
1-\frac{m^2R}{a^2d}\frac{\varphi_d(R)}{\varphi_d'(R)}=\frac{mR}{ad}\frac{K_{\frac{d}{2}+1}\left(\frac{mR}{a}\right)}{K_{\frac{d}{2}}\left(\frac{mR}{a}\right)}\to 1\quad\text{ as $R\to0$}\ , 
\end{align}
and is nondecreasing (for $R>0$). Since $v(r)$ and $v'(r)r^{d-1}$ are both strictly increasing, we conclude that $c(0)= 1$, and $c(R)$ is strictly decreasing. Therefore, there exists a positive radius $R_{\star}>0$ such that $c(R_{\star})=0$.
The pulled phase is characterised by the absence of condensation of charge on the surface: $c(R)=0$  for $R\geq R_{\star}$. Hence the singular component on the surface is nonnegative for all $R\geq0$.
\par
We now analyse the bulk. This amounts to showing that $f(r)$ in~\eqref{eq:bulk} is nonnegative. First, note that  $a^2(r^{d-1}v'(r))'$ is  positive by assumption. Since $v(r)$ is strictly increasing, it is enough to show that $\mu(R)-v(R)\geq0$ to imply that $\mu(R)-v(r)\geq0$ for all $r\leq R$. In the pushed phase
\be
\mu(R)-v(R)=- \frac{\varphi_d(R)}{\varphi_d'(R)}\left(v'(R)+\frac{c(R)}{a^2R^{d-1}}\right)\ .
\ee
The ratio $\frac{\varphi_d(R)}{\varphi_d'(R)}$ is negative while $c(R)$ and $v'(R)$ are both positive. Therefore, $f(r)\geq0$ for $r\leq R\wedge R_{\star}$. 
\end{remark}
Before embarking in the calculations leading to the explicit formulae for $\mu(R)$ and $c(R)$, we derive the universal formula of the free energy (Theorem~\ref{thm:Yukawa}) assuming Theorem~\ref{thm:EM_Yukawa}.
\begin{proof}[of Theorem~\ref{thm:Yukawa}]  First, note that for $R\geq R_{\star}$,
\begin{equation}
\rho_R=\rho_{R_{\star}}\quad\Rightarrow\quad F(R)=\EE[\rho_R]-\EE[\rho_{R_{\star}}]=0\ .
\label{eq:Fpulled}
\end{equation}
For $R< R_{\star}$, we compute $\EE[\rho_R]$ by inserting the explicit form of the minimiser $\rho_R$ in the functional.
This calculation is made possible by using the explicit formulae~\eqref{eq:sigmaR}--\eqref{eq:mu_Yuk} for $\rho_R$ and, crucially, a shell theorem for Yukawa interaction (see Lemma~\ref{lem:shell_thm_Yukawa} below). Eventually, we find for $R<R_{\star}$:
\begin{equation}
F(R)=\frac{1}{2} \int_{R}^{R_{\star}} \frac{c(r)^2}{a^2r^{d-1}}\de r>0\ .
\label{eq:Fpushed}
\end{equation}
Combining~\eqref{eq:Fpulled} and~\eqref{eq:Fpushed} we obtain the claimed formula~\eqref{eq:simple_Y}.
\par
To prove the jump in the third derivative, we remark again that $c(R)$ is continuous for $R\geq0$ and differentiable at least twice for $R>0$ and $R\neq R_{\star}$. Moreover, $R_{\star}>0$, and  $c(R_{\star})=0$. From the explicit formula~\eqref{eq:simple_Y}:
\begin{align}
\lim_{R\uparrow  R_{\star}}F(R)&=\frac{1}{2} \int_{R_{\star}}^{R_{\star}} \frac{c(r)^2}{a^2r^{d-1}}\de r=0\ ,\\
\lim_{R\uparrow  R_{\star}}F'(R)&=-\frac{1}{2a^2}\frac{c(R_{\star})^2}{R_{\star}^{d-1}}=0\ ,\\
\lim_{R\uparrow  R_{\star}}F''(R)&=-\frac{1}{2a^2}\left(\frac{2c(R_{\star})c'(R_{\star})R_{\star}^{d-1}-(d-1)c(R_{\star})^2R_{\star}^{d-2}}{R_{\star}^{2(d-1)}}\right)=0\ .
\end{align}
On the other hand, 
\begin{align}
\lim_{R\uparrow  R_{\star}}F'''(R)&=-\frac{1}{a^2}\frac{c'(R_{\star})^2}{R_{\star}^{d-1}}<0\ ,
\end{align}
since $c(R)$ is strictly decreasing in the pushed phase. \qed
\end{proof}

\subsection{Chemical potential and excess charge}
The chemical potential $\mu(R)$ and the excess charge $c(R)$ are fixed by the normalisation of $\rho_R$ and the E-L conditions.  

\par
Assume $R\leq R_{\star}$ (note that at this stage $R_{\star}$ is not known). We show here that in the pushed phase the chemical potential and the excess charge are solution of the following linear system
\be
\begin{cases}
& \dfrac{\varphi_d'(R)}{\varphi_d(R)}\mu(R)+\dfrac{1}{a^2R^{d-1}}c(R) =\dfrac{\varphi_d'(R)}{\varphi_d(R)}v(R)-v'(R)\ , \\\\
&\dfrac{m^2R^d}{d}\mu(R)+c(R) =1-a^2v'(R)R^{d-1}+m^2\int\limits_{0}^Rv(r)r^{d-1}\de r\ ,
\label{eq:mu_c_Yuk}
\end{cases}
\ee
from which the  expressions for $\mu(R)$ and $c(R)$ in~\eqref{eq:c_Yuk}--\eqref{eq:mu_Yuk} follow. Note that, for all $R$, 
\begin{align}
\det
\begin{pmatrix}
   \dfrac{\varphi_d'(R)}{\varphi_d(R)}      & \dfrac{1}{a^2R^{d-1}} \\\\
   \dfrac{m^2R^d}{d}     & 1
\end{pmatrix}
=-\frac{m^2R}{a^2d}\dfrac{K_{\frac{d}{2}+1}\left(\frac{mR}{a}\right)}{K_{\frac{d}{2}-1}\left(\frac{mR}{a}\right)}\neq0\ ,
\end{align}
therefore the solution of the linear system~\eqref{eq:mu_c_Yuk} is unique.
\par
The normalisation condition is
\be
\int\de\rho_R(x)=1\,\,\Rightarrow\,\, \frac{m^2R^d}{d}\mu(R)-m^2\int_0^Rv(r)r^{d-1}\de r+a^2R^{d-1}v'(R)+c(R)=1\ .\label{normcond}
\ee
\par
The E-L condition in the support of $\rho_R$ is
\be
\int\varphi_d(|x-y|)\de\rho_R(y)=\mu(R)-v(|x|),\quad\text{for $|x|\leq R$}\ .
\label{eq:E-L_partial}
\ee
At this stage, we need an analogue of the electrostatic shell theorem to perform the angular integration in~\eqref{eq:E-L_partial}. For Yukawa interaction the potential outside a uniformly `charged' sphere is the same as that generated by a point charge at the centre of the sphere with a dressed charge (dependent on the radius of the sphere); inside the spherical shell the potential is not constant. The precise statement is a `shell theorem' for screened Coulomb interaction.
\begin{lemma}[Shell integration formula] 
\label{lem:shell_thm_Yukawa}
\be
\frac{1}{\Omega_d}\int\limits_{|y|=r}\varphi_d\left(|x-y|\right)\de S(y)=\psi_d\left(\min\left\{|x|,r\right\}\right)\varphi_d\left(\max\left\{|x|,r\right\}\right)\ ,
\label{eq:shell_thm_Yukawa}
\ee
where $\de S$ is the $(d-1)$-dimensional surface measure, with
\be
\psi_d(r)=\left(\frac{2a}{m r}\right)^{\frac{d}{2}-1}\Gamma\left(\frac{d}{2}\right)I_{\frac{d}{2}-1}\left(\frac{m r}{a}\right)\ .
\ee
($I_{\nu}$ denotes the modified Bessel function of the first kind.) 
\begin{proof}
See Appendix~\ref{app:lemmas}.
\qed
\end{proof}
\end{lemma}
\begin{remark} The integration formula~\eqref{eq:shell_thm_Yukawa} has already appeared in disguised form in the literature, see~\cite[Theorem 6]{Duffin71} or~\cite[Theorem 4.2]{Rasila16}. In dimension $d=1,2$, and $3$:
\be
\psi_1(r)=\cosh(mr/a),\quad \psi_2(r)=I_0(mr/a),
\quad \psi_3(r)=\frac{\sinh(mr/a)}{mr/a}\ .
\ee
It is interesting to compare~\eqref{eq:shell_thm_Yukawa} with the classical shell theorem when $\varphi_d$ is the $d$-dim Coulomb potential
 \be
\frac{1}{\Omega_d}\int\limits_{|y|=r}\varphi_d\left(|x-y|\right) \de S(y) =\varphi_d\left(\max\left\{|x|,r\right\}\right)\ .
\label{eq:Newton}
\ee
\end{remark}
Using~\eqref{eq:shell_thm_Yukawa}, we can perform the angular integration in~\eqref{eq:E-L_partial}
\begin{align}
\varphi_d(z)\int_0^{z}\psi_d(r)f(r)\de r+\psi_d(z)\int_{z}^R\varphi_d(r)f(r)\de r +c(R)\varphi_d(R)\psi_d(z)=\mu(R)-v(z)\ ,
\end{align}
where we set $z=|x|\leq R$.

Differentiating with respect to $z$,
\begin{align}
\varphi_d'(z)\int_0^{z}\psi_d(r)f(r)\de r
+\psi_d'(z)\int_{z}^R\varphi_d(r)f(r)\de r 
+c(R)\varphi_d(R)\psi_d'(z)+v'(z)=0\ .
\label{eq:E-L_step}
\end{align}
Using integration by parts and the properties of $\psi_d$ and $\varphi_d$ we find
\begin{align}
\nonumber \int_0^z\psi_d(r)f(r)\de r&=a^2\psi_d(z)z^{d-1}v'(z)+a^2\psi_d'(z)z^{d-1}(\mu(R)-v(z))\ ,\\
\int_z^R\varphi_d(r)f(r)\de r&=a^2\varphi_d(R)R^{d-1}v'(R)+a^2\varphi_d'(R)R^{d-1}(\mu(R)-v(R))\nonumber\\
&-a^2\varphi_d(z)z^{d-1}v'(z)-a^2\varphi_d'(z)z^{d-1}(\mu(R)-v(z))\ .
\end{align}
Therefore, Eq.~\eqref{eq:E-L_step} reads
\begin{align}
\nonumber a^2\psi_d(z)z^{d-1}v'(z)+a^2\varphi_d(R)R^{d-1}v'(R)+a^2\varphi_d'(R)R^{d-1}(\mu(R)-v(R))\\-a^2\varphi_d(z)z^{d-1}v'(z)+c(R)\varphi(R)\psi_d'(z)+v'(z)=0\ ,
\label{eq:E-L_step2}
\end{align}
which can be simplified using the following lemma.
\begin{lemma} For $z\neq0$,
\be
\varphi_d'(z)\psi_d(z)-\varphi_d(z)\psi_d'(z)=-\frac{1}{a^{2}}\frac{1}{z^{d-1}}\ .
\label{eq:Abel}
\ee
\end{lemma}
\begin{proof}
 $I_{\nu}$ and $K_{\nu}$ are solutions of the modified Bessel equation
\be
z^2\frac{\de^2w}{\de z^2}+z\frac{\de w}{\de z}-(z^2+\nu^2)w=0\ .
\ee
Equation~\eqref{eq:Abel}  is a straightforward application of Abel's identity 
to the Wronskian of $I_{\nu}$ and $K_{\nu}$, see~\cite[Eq.~10.28.2]{NIST}.
\qed
\end{proof}
Using the above lemma, Eq.~\eqref{eq:E-L_step} simplifies as
\begin{align}
a^2\varphi_d(R)R^{d-1}v'(R)\psi_d'(z)+a^2\varphi_d'(R)R^{d-1}(\mu(R)-v(R))\psi_d'(z)+c(R)\varphi(R)\psi_d'(z)=0\ .
\label{eq:E-L_step3}
\end{align}
The above equality must be true for all $z$. Therefore, 
the E-L equation can be eventually written as
\begin{equation}
\frac{c(R)}{R^{d-1}}+a^2\frac{\varphi_d'(R)}{\varphi_d(R)}(\mu(R)-v(R))+a^2v'(R)=0\ .
\end{equation}
The above relation between $\mu(R)$ and $c(R)$ is the first equation in~\eqref{eq:mu_c_Yuk}.
\par
We finally get the explicit formulae
\begin{align}
\mu(R)&=\frac{v(R)-\frac{1}{a^2R^{d-1}}\frac{\varphi_d(R)}{\varphi_d'(R)}\left(1-m^2\int_0^Rv(r)r^{d-1}\de r\right)}{1-\frac{m^2R}{a^2d}\frac{\varphi_d(R)}{\varphi_d'(R)}}\ ,
\label{eq:mu_Yukawa}\\
 c(R)&=\frac{1-\left(a^2-\frac{m^2R}{d}\frac{\varphi_d(R)}{\varphi_d'(R)}\right)v'(R)R^{d-1}-\frac{m^2R^d}{d}v(R)+m^2\int_0^Rr^{d-1}v(r)\de r}{1-\frac{m^2R}{a^2d}\frac{\varphi_d(R)}{\varphi_d'(R)}}\ .
\label{eq:c_Yukawa}
\end{align}
At this stage, what would remain to do is just checking the positivity of $\rho_R$. We did this in Remark~\ref{rmk:positivity}. This concludes the proof of the theorem. 
\begin{figure}
\centering
  \includegraphics[width=1\columnwidth]{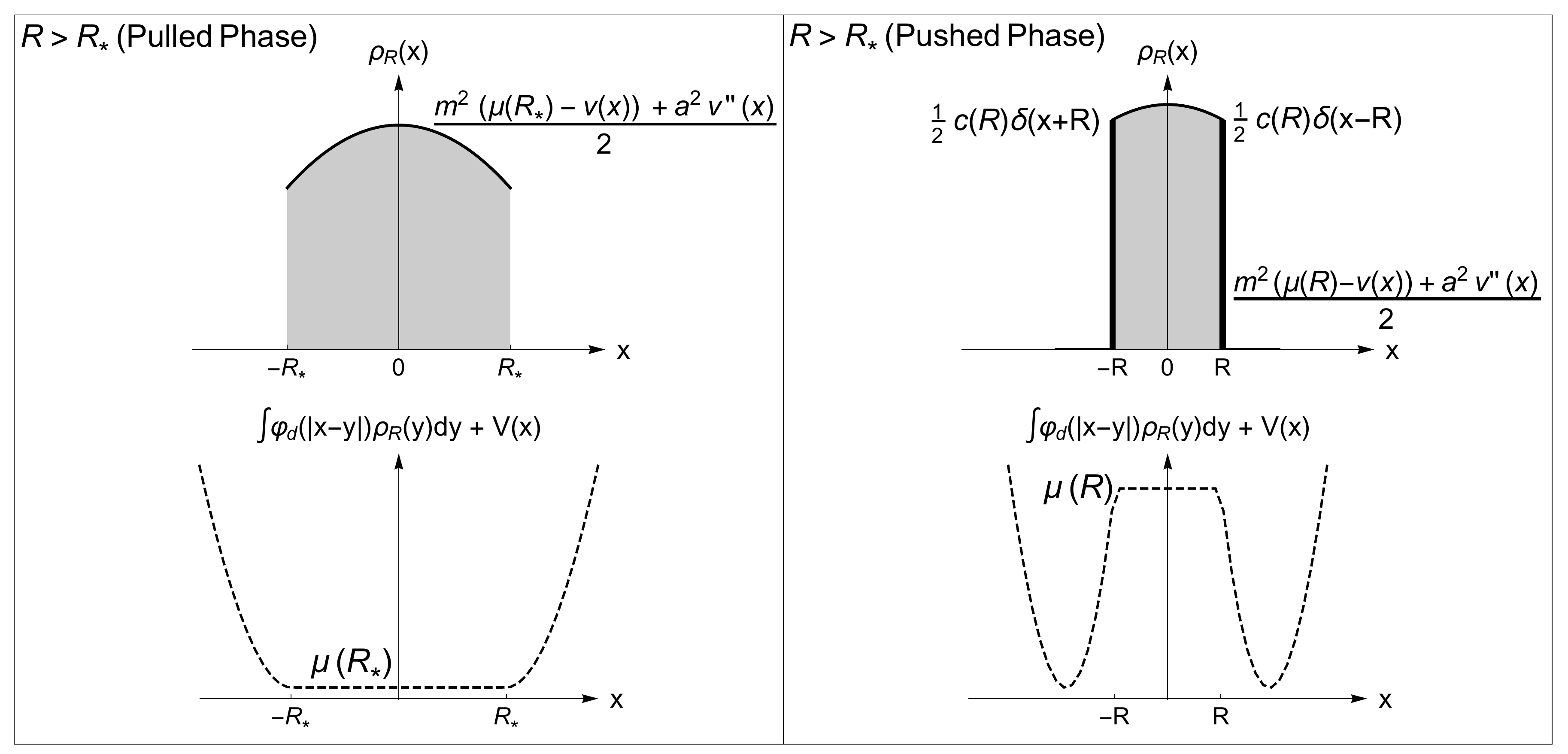}
\caption{Top: Equilibrium measure of a one-dimensional ($d=1$) Yukawa gas in the pulled and pushed phases. Here the gas is confined in a quadratic potential $v(r)=r^2/2$, and $a=m=1/2$. Bottom: energy density (on a convenient scale) corresponding to the equilibrium measures. The energy density is constant, and equal to the chemical potential $\mu$, in the support of the equilibrium measure. }
\label{fig:density_energy}       
\end{figure}

\begin{acknowledgements}
The research of FDC is supported by ERC Advanced Grant 669306.
PV acknowledges the stimulating research environment provided by the EPSRC Centre for Doctoral Training in Cross-Disciplinary Approaches to Non-Equilibrium Systems (CANES, EP/L015854/1).
ML was supported by Cohesion and Development Fund 2007--2013 - APQ Research Puglia Region ``Regional program supporting smart specialization and social and environmental sustainability - FutureInResearch.'' 
PF was partially supported by by Istituto Nazionale di Fisica Nucleare (INFN) through the project ``QUANTUM.''
FDC, PF and ML were partially supported by the Italian National Group of Mathematical Physics (GNFM-INdAM).
The authors would like to thank S. N. Majumdar and  Y. V. Fyodorov for useful discussions, and G. Schehr and Y. Levin for  remarks on the first version of the paper. We also thank E. Katzav for drawing our attention to the similarity between some of our calculations and the $J$-integral of linear elastic fracture mechanics theory.
\end{acknowledgements}
\appendix
\section{Yukawa interaction}
\label{app:Yukawa}
Writing the distributional identity~\eqref{eq:distr} in Fourier coordinates
\ben
\widehat{\Phi}_d(k)=\frac{\Omega_d}{a^2|k|^2+m^2}\quad\Rightarrow\quad
\Phi_d(x)=\frac{\Omega_d}{(2\pi)^d}\int_{\R^d}\frac{e^{-\mathrm{i}kx}}{a^2|k|^2+m^2}\de k\ .
\label{eq:kernel_Fourier}
\een
One way to invert the Fourier transform is by writing
\ben
\Phi_d(x)=\frac{\Omega_d}{(2\pi)^d}\int_{\R^d}\de k\int_0^{+\infty}\de te^{-\mathrm{i}kx}e^{-(a^2|k|^2+m^2)t}\ ,
\een
and work out first the Gaussian integral in $k$:
\ben
\Phi_d(x)=\frac{\Omega_d}{(4\pi a^2)^\frac{d}{2}}\int_0^{+\infty}\frac{e^{-m^2t-\frac{|x|^2}{4a^2t}}}{t^\frac{d}{2}}\de t\ .
\een
One recognises the integral representation of the Bessel function of the second kind~\cite[Eq.~10.32.10]{NIST} 
\ben
\Phi_d(x)=\varphi_d(|x|)\ , \qquad
\varphi_d(r)=\frac{1}{a^2 2^{\frac{d}{2}-1}}\frac{1}{\Gamma\left(\frac{d}{2}\right)}\left(\frac{m}{ar}\right)^{\frac{d}{2}-1}K_{\frac{d}{2}-1}(mr/a)\ .
\een
In dimension $d=1,2$, and $3$ we have the familiar expressions
\ben
\varphi_1(r)=\frac{e^{-mr/a}}{am},\quad \varphi_2(r)=\frac{K_0(mr/a)}{a^2}\ ,
\quad \varphi_3(r)=\frac{e^{-mr/a}}{a^2r}\ .
\een

\section{Proofs of lemmas}
\label{app:lemmas}

\begin{dimostrazione}{Eq.~\eqref{eq:formula:Cheb}}
We wish to prove the identity
\begin{equation}
-\log |x-y|=\log 2+\sum_{n\geq 1}\frac{2}{n}T_n(x)T_n(y)\qquad |x|\leq 1,|y|\leq 1,x\neq y\ .\label{identitytoprove}
\end{equation}
The identity follows from $T_n(\cos\theta)=\cos(n\theta)$. Calling $X=\arccos x$ and $Y=\arccos y$, we have to evaluate the sum $S=\sum_{n\geq 1}\frac{\cos(n X)\cos(n Y)}{n}$ on the right hand side . Using the trigonometric identity
\begin{equation}
2\cos\alpha\cos\beta=\cos(\alpha-\beta)+\cos(\alpha+\beta)\ ,
\end{equation}
we have
\begin{align}
\nonumber S &=\frac{1}{2}\left[\sum_{n\geq 1}\frac{\cos(n(X-Y))}{n}+\sum_{n\geq 1}\frac{\cos(n(X+Y))}{n}\right]=\frac{1}{2}\mathrm{Re}\left[\sum_{n\geq 1}\frac{\exp(\mathrm{i}n(X-Y))}{n}+\sum_{n\geq 1}\frac{\exp(\mathrm{i}n(X+Y))}{n}\right]\\
 &=-\frac{1}{2}\mathrm{Re}\left[\log(1-e^{\mathrm{i}(X-Y)})+\log(1-e^{\mathrm{i}(X+Y)})\right]=-\frac{1}{4}\left[\log (2-2 \cos (X-Y))+\log (2-2 \cos (X+Y))\right]\ .
\end{align}
Therefore, all we need to compute is
\begin{align}
\nonumber \cos\left(\arccos x\pm\arccos y\right) &=\cos(\arccos x)\cos(\arccos y)\mp \sin(\arccos x)\sin(\arccos y)\\
&=xy\mp \sqrt{1-x^2}\sqrt{1-y^2}\ ,
\end{align}
using the standard trigonometric addition formula. After simplifications
\begin{equation}
S=-\frac{1}{4}\log[4 (x-y)^2]=-\frac{1}{2}\left(\log 2+\log |x-y|\right)\ .
\end{equation}
Substituting in the r.h.s. of~\eqref{identitytoprove}, we obtain the claim. 
\end{dimostrazione}

\begin{dimostrazione}{Lemma~\ref{lem:electrostatic_log-gas}}
The case $|x|<1$ is an immediate application of the identity~\eqref{eq:formula:Cheb} and the orthogonality relation~\eqref{eq:Cheby_orth} of the Chebyshev polynomials. 
\par
We consider now the case $x\geq1$ (the case $x\leq-1$ is similar).
When $x\geq1$, $\cosh^{-1}(x)=\log(x+\sqrt{x^2-1})$, so that the identity we wish to prove for $n\geq1$ is
\ben
I_n(x)=-\int_{-1}^1 \log (x-y)\frac{T_n(y)}{\pi\sqrt{1-y^2}}\de y=\frac{1}{n}(x+\sqrt{x^2-1})^{-n}\qquad x\geq 1 \ .\label{identity}
\een
Set $x=(1+t^2)/(2t)$, for $t>1$. Multiplying~\eqref{identity} by $z^n$ and summing for $n\geq 1$
\ben
\sum_{n\geq 1}I_n\left(\frac{1+t^2}{2t}\right)z^n =\sum_{n\geq 1}\frac{z^n}{n t^n}=-\log\left(1-\frac{z}{t}\right)\ .
\label{eq:identity2app}
\een
After elementary manipulations we find ($n\geq 1$)
\ben
I_n\left(\frac{1+t^2}{2t}\right) =-J_n(t)\ ,\quad \text{where}\quad J_n(t)=\int_{-1}^1 \log \left(1+t^2-2ty\right)\frac{T_n(y)}{\pi\sqrt{1-y^2}}\de y\ ,
\een
so that
\ben
\sum_{n\geq 1}I_n\left(\frac{1+t^2}{2t}\right)z^n =-\sum_{n\geq 1}J_n\left(t\right)z^n\ .
\een
Therefore the identity is established if we show that 
\begin{equation}
\sum_{n\geq 1}J_n\left(t\right)z^n=\log\left(1-\frac{z}{t}\right)\qquad |z/t|<1\ .\label{toshow}
\end{equation}
The left hand-side is 
\ben
\sum_{n\geq 1}J_n\left(t\right)z^n=\int_{-1}^1\log \left(1+t^2-2ty\right)\frac{\sum_{n\geq 1}T_n(y)z^n}{\pi\sqrt{1-y^2}}\de y= \int_{-1}^1 \log \left(1+t^2-2ty\right)\frac{\sum_{n\geq 0}T_n(y)z^n-T_0(y)}{\pi\sqrt{1-y^2}}\de y\ .
\een
Using 
\ben
\int_{-1}^1 \log \left(1+t^2-2ty\right)\frac{T_0(y)}{\pi\sqrt{1-y^2}}\de y=2\log t\ ,
\een
and the generating function of the Chebyshev polynomials~\cite[Eq.~18.12.8]{NIST}
\ben
\sum_{n\geq 0}T_n(y)z^n=\frac{1-zy}{1-2zy+z^2}\ ,
\een
we get
\ben
 \sum_{n\geq 1}J_n\left(t\right)z^n=\int_{-1}^1 \log \left(1+t^2-2ty\right)\frac{1}{\pi\sqrt{1-y^2}}\frac{1-zy}{1-2zy+z^2}\de y -2\log t\ .\label{beforederivative}
\een
It is convenient to evaluate the $t$-derivative 
\be
\frac{\partial}{\partial t} \sum_{n\geq 1}J_n\left(t\right)z^n=-\frac{2}{t}+\int_{-1}^1 \frac{2t-2y}{1+t^2-2ty}\frac{1}{\pi\sqrt{1-y^2}}\frac{1-zy}{1-2zy+z^2}\de y \ .
\label{tder}
\ee
The last integral  can be evaluated with a partial fraction expansion and we find (for $t>1$ and $-1<z<1$)

\begin{equation}
\frac{\partial}{\partial t} \sum_{n\geq 1}J_n\left(t\right)z^n=-\frac{1}{t}+\frac{1}{t-z}\ .
\end{equation}
Comparing with the right hand-side of~\eqref{toshow}, we conclude that 
\begin{equation}
\sum_{n\geq 1}J_n\left(t\right)z^n=\log\left(1-\frac{z}{t}\right)+k(z)\ .
\end{equation}
To fix the constant $k(z)$, we may appeal to the large-$t$ behaviour from the right-hand side of~\eqref{tder}
\begin{equation}
\lim_{t\to\infty} \sum_{n\geq 1}J_n\left(t\right)z^n=0\ .
\end{equation}
Therefore, $k(z)=0$ and the proof of~\eqref{eq:identity2app} is complete.
\end{dimostrazione}
\begin{dimostrazione}{Lemma~\ref{lem:shell_thm_Yukawa}}
Using the integral representation of the Bessel function $K_{\frac{d}{2}-1}$ and writing $| x- y|^2=| x|^2+| y|^2-2 (x\cdot y) $, the integral becomes
\begin{equation}
\frac{1}{(4\pi a^2)^{d/2}}\int_0^\infty\frac{\de t}{t^{d/2}}e^{-m^2 t-\frac{1}{4 a^2 t}(| x|^2+r^2)}\int_{| y|=r}\de y\ \exp\left({\frac{1}{2a^2 t} x\cdot y}\right)\ .
\end{equation}
 
 Noting that the problem is rotationally invariant, without loss of generality we can assume that $ x$ has all zero components but one, $ x=(| x|,0,\ldots,0)$. Introducing $d$-dimensional spherical coordinates
\begin{align}
y_1 &=r\cos\phi_1\nonumber\\
y_2 &= r\sin\phi_1\cos\phi_2\nonumber\\
y_2 &=r\sin\phi_1\sin\phi_2\cos\phi_3\nonumber\\
 &\vdots\,\nonumber\\
y_d &= r\sin\phi_1\cdots\sin\phi_{d-2}\sin\phi_{d-1}\quad(0 \leq\phi_1\leq2\pi,\;\;  0\leq\phi_j\leq\pi, \,\text{for $2\leq j\leq d$})\ , 
\end{align} 
the integral over the sphere becomes
 \begin{align}
& \int_{| y|=r}\de y\ \exp\left[{\frac{1}{2a^2 t} x\cdot y}\right]=C\int_0^\pi \de\phi_1 \exp\left[{\frac{| x|r}{2a^2 t}\cos\phi_1}\right](\sin\phi_1)^{d-2}\ ,\\
 &\text{with}\quad C=\int_0^{2\pi}d\phi_{d-1}\prod_{j=2}^{d-2}\int_0^\pi \de\phi_j \left(\sin(\phi_j)\right)^{d-1-j}\ .
 \end{align}
 The $\phi_1$-integral corresponds to the integral representation of the modified Bessel function $I_{\frac{d}{2}-1}$ and, after simplifications, the original integral over the sphere~\eqref{eq:shell_thm_Yukawa} is shown to be equivalent to the one-dimensional integral
 \begin{equation}
 \frac{(| x|r)^{1-d/2}}{2 a^2}\int_0^\infty\frac{\de t}{t}e^{-m^2 t-\frac{1}{4 a^2 t}(| x|^2+r^2)}I_{d/2-1}\left(\frac{| x|r}{2a^2 t}\right)\ .
 \end{equation}
After the change of variables $z=\frac{| x|r}{2a^2 t}$, we can use the following cute identity (a version of~\cite[Eq. 6.635.3]{Gradshteyn})
 \begin{equation}
\int_0^{\infty } \frac{\de z}{z}I_{\nu}(z) \exp \left(-\frac{\eta  \xi }{2 z}-\frac{z}{2}  \left(\frac{\eta }{\xi }+\frac{\xi }{\eta }\right)\right)=
 \begin{cases}
 2 K_{\nu}(\eta ) I_{\nu}(\xi ) &\mbox{ for }\eta\geq\xi\\
 2 I_{\nu}(\eta ) K_{\nu}(\xi )&\mbox{ for }\eta<\xi\\
 \end{cases}\ ,
 \end{equation}
 valid for ${\nu}\geq 0$ and $\xi,\eta>0$. This concludes the proof of the Lemma.
 \end{dimostrazione}

\end{document}